\documentclass[letterpaper,pre,reprint,aps,twocolumn,superscriptaddress,nofootinbib]{revtex4-2}
\makeatletter
\def\selectlanguage#1{}
\makeatother
\usepackage[T1]{fontenc}
\usepackage[utf8]{inputenc}
\usepackage[table]{xcolor}
\usepackage{amsmath, amssymb, esint, listings, afterpage, float, soul, physics, textgreek, tikz, mathtools, multirow, lipsum, lmodern, blindtext, bigints, xurl, mathtools, algorithm, algpseudocode, dcolumn, array, graphicx}

\usepackage[toc,page]{appendix}
\usetikzlibrary{quantikz2}
\usepackage{diagbox}
\usepackage{tikz}
\usepackage{amsthm}

\newtheorem{theorem}{Theorem}

\usepackage[colorlinks = true,
            linkcolor = blue,
            urlcolor  = blue,
            citecolor = blue,
            anchorcolor = blue]{hyperref}
            
\usepackage[most]{tcolorbox}
\usepackage[table]{xcolor}
\usepackage{orcidlink}
\usetikzlibrary{calc,positioning, arrows}
\usepackage{geometry}
\tcbuselibrary{most}

\tcbset{
    myalgobox/.style={
        colback=gray!5,
        colframe=gray!50,
        coltitle=white,
        fonttitle=\bfseries,
        title filled,
        colbacktitle=gray!80,
        arc=2pt,
        boxrule=0.8pt,
        left=2pt,
        right=2pt,
        top=4pt,
        bottom=4pt
    }
}

\newcommand{\lrtb}[1]{\left[#1\right]}
\newcommand{\lrfb}[1]{\left(#1\right)}

\newtheorem{thm}{Theorem}[section]

\newtheorem{lemma}[thm]{Lemma}
\newcommand{\diag}{\text{diag}}

\newcommand{\Ftel}{F^{\text{tel}}_{\text{avg}}}

\begin{document}

\title{Enhancing the teleportation fidelity of a quantum network using purification}
\author{Soumit Roy\,\orcidlink{0009-0002-4579-0491}}
\email{soumit.roy@research.iiit.ac.in}
\affiliation{Centre for Quantum Science and Technology, International Institute of Information Technology, Hyderabad, Hyderabad – 500032, Telangana, India.}

\author{Md Sohel Mondal\,\orcidlink{0009-0006-2915-2309}}
\email{mondalsohelmd@gmail.com}
\affiliation{Department of Physics and Center of excellence in Quantum information, computation, science and technology, Indian Institute of Technology Bombay, Mumbai, Maharashtra 400076, India}

\author{Siddhartha Santra\,\orcidlink{0000-0001-6695-3929}}
\email{santra@iitb.ac.in}
\affiliation{Department of Physics and Center of excellence in Quantum information, computation, science and technology, Indian Institute of Technology Bombay, Mumbai, Maharashtra 400076, India}

\author{Indranil Chakrabarty\,\orcidlink{0009-0001-0415-0431}}
\email{indranil.chakrabarty@iiit.ac.in}
\affiliation{Centre for Quantum Science and Technology, International Institute of Information Technology, Hyderabad, Hyderabad – 500032, Telangana, India.}
\affiliation{Center for Security, Theory and Algorithmic Research, International Institute of Information Technology, Hyderabad, Hyderabad – 500032, Telangana, India.}

\date{\today}

\begin{abstract}

Complex quantum networks can support a diverse set of long-range entanglement distribution schemes ranging from linear repeater protocols to multipath entanglement purification strategies. As a result, a network's resourcefulness, that is its ability to facilitate quantum communication, depends on the deployed distribution scheme. In this work, we analyse and compare the resourcefulness of quantum networks across a broad range of network topologies, including both regular and random networks, under two distinct entanglement distribution schemes. The first relies on entanglement swapping along a single path connecting a source–target pair, while the second exploits entanglement purification using multiple paths between the same source and target nodes. The resourcefulness of the network is quantified using a recently described metric \cite{mylavarapu2024teleportation} that averages over the maximum teleportation fidelity between arbitrary source-target pairs in the network. We present algorithms for estimating this metric under constraints of edge-usage and ordering of paths. Our results not only demonstrate the sensitivity of the average maximum teleportation fidelity metric to the choice of entanglement distribution protocol, but also highlight the significant improvements enabled by network purification schemes. In particular, purification-based approaches can substantially enhance average teleportation fidelity, thereby improving the overall teleportation capability of quantum networks.

\end{abstract}

\maketitle

\section{Introduction}
\noindent Entanglement-based quantum repeater networks \cite{aspelmeyer2003long,bedingtonprogress,briegel1998quantum,duan2001long,gour2004remote,sazim2013study} are foundational building blocks in establishing the quantum internet in future, which can eventually carry out various advanced distributed tasks  \cite{bennett1993teleporting,chakrabarty2010teleportation,chakrabarty2011deletion,adhikari2008quantum,sohail2023teleportation,agrawal2002probabilistic,singh2024controlled,adhikari2010probabilistic,pati2000minimum,bennett2005remote,bennett1992communication,patro2017non} between any pair of nodes. Though there are some universal limitations \cite{neumann2025no,bauml2015limitations,das2021practically} and milestones to be achieved experimentally,  it is believed that before we can have a full-fledged quantum network with advanced memory, many theoretical considerations need to be looked at to understand the network capability, not only from the source-target perspective but also globally. While large-scale classical networks \cite{albert2002statistical,watts1998collective,pastor2015epidemic,boccaletti2014structure,ji2023signal,hens2019spatiotemporal,ghavasieh2024diversity,meena2023emergent,moore2020predicting,sharma2020emergence} have been studied globally, large-scale quantum networks mostly remain unexplored.\\
\indent Quantum teleportation \cite{bennett1993teleporting,sohail2023teleportation,agrawal2002probabilistic,agrawal2010task} being  one of the fundamental protocol in transferring quantum information from one location to another,  is not restricted  to a single source and receiver; and can be implemented  across multiple sources and targets in a network \cite{mylavarapu2023entanglement,alimuddin2026entanglement, ghosal2025repeater, maiti2026}.
A state being called as a resource for teleportation goes back to the notion of quantum advantage for teleportation using a two-qubit mixed state as the shared state between the sender and receiver \cite{bennett1993teleporting,horodecki1996teleportation}. This gives rise to a resource theory framework for quantum teleportation, as we can find entangled states that are not useful for teleportation \cite{horodecki1996teleportation,chakrabarty2010teleportation,chakrabarty2011deletion}, along with the creation of a witness for teleportation \cite{ganguly2011entanglement}. Interestingly, this notion of advantage can be extrapolated in the context of quantum network \cite{mylavarapu2024teleportation,mylavarapu2023entanglement,roy2026teleportation}.  \\  
\indent One of the key features of the repeater based quantum technologies is the process of entanglement swapping \cite{bose1998multiparticle,bose1999purification, pan1998experimental,ghosal2025repeater,alimuddin2026entanglement}, which enables us to establish long-range entanglement links from short-range entanglement links between the nodes  \cite{briegel1998quantum} and thus connecting the short-range entanglement and teleportation fidelity with long-range teleportation fidelity \cite{sazim2013study, gour2004remote, guerra2025entanglement}. Entanglement swapping in repeater technology becomes handful for distant ground-based communication, where one normally has to send the qubit through fibres, which are prone to noise \cite{azuma2023quantum} (whether the transfer is done node to node or by heralded photon generation). Theoretically, the performance of a quantum repeater has been studied both analytically \cite{hartman2007, filip2018, khatri2019, brand2020, li2021, Goodenough2025,azuma2023quantum} and through simulations \cite{ferreira2021optimizing, Avis2023}. Experimental generation of entanglement between distant matter qubits located at remote nodes has been demonstrated on various physical platforms like trapped ions \cite{stephenson2020, Krutyanskiy2019, krutyanskiy2023, Krutyanskiy2023may, Schupp2021}, colour centres \cite{Bernien2013, Humphreys2018}, rare earth ions \cite{Liu2021, Lago-Rivera2021, Ruskuc2025} and atomic ensemble memories \cite{Chin2007, Yuan2008, Yu2020}.\\
\indent The average teleportation fidelity \cite{mylavarapu2024teleportation} is a global measure of the resourcefulness of a quantum network as a whole and can be used to characterize its capability for transferring quantum information, both in the presence and absence of loops. For computing the average teleportation fidelity of the network, one needs to consider all pairs of source and target. In particular, for multiple paths between a source and target, the best path's contribution is considered in the average teleportation fidelity. However, networks with loops provide alternative paths from the source to the target. Here in this article, we want to see whether the multipath entanglement purification (MPEP) method gives us a better alternative than choosing the best path by increasing the average fidelity of the network. The MPEP method creates a single path from the multiple paths between a source and a target after purification. Interestingly, we see that indeed for regular closed topologies and for large complex networks Erd\H{o}s-R\'enyi's network (ERN) purification of multiple paths between a source and target nodes to a single path can enhance the average teleportation fidelity of the network. Multipath purification can make the network more useful in terms of its ability for teleportation.\\
\indent In this article, we propose two algorithms for finding the average of teleportation fidelity via the multipath purification method. The two algorithms are based on two different scenarios: the output for two different cases, where one edge can be used only once, and one edge can be used the maximum number of times, assuming there is an ensemble of entangled states for each edge. We have also adopted two different strategies, namely the shortest path last Strategy (SPL) and the shortest path first Strategy (SPF) \cite{mondal2024} for each of these algorithms. When the network has a limited number of entangled states as resources between two nodes, the edges can not be used multiple times for different paths. We also show that the calculation of $\Ftel$ with no purification situation can be outperformed by using a single edge only once, along with a suitable purification approach for certain topologies. We consider regular topologies with loops (Ring Network (RN), Complete Graph Network (CGN), Triangular Lattice Network (TLN), Square Lattice Network (SLN)) and one of the irregular topologies (Erd\H{o}s-R\'enyi Network (ERN)). Our findings establish that the MPEP strategy is always better than the no purification strategy (better than the best path method in \cite{mylavarapu2024teleportation}) and can enhance global network metrics like average teleportation fidelity.\\
\indent In section \ref{sec2}, we briefly discuss existing concepts like quantum network model, teleportation fidelity, the average teleportation fidelity of a global network, and Deutsch's protocol for entanglement purification. In section \ref{sec3}, we elaborate on the strategies for MPEP and propose the algorithm to calculate the average of teleportation fidelity via the MPEP methods. This section also consists of finding the optimal purification strategy for maximizing the output teleportation fidelity. In section \ref{sec4}, we present the analytical derivation of the average teleportation fidelity and the results associated with it for regular topologies like Ring Network (RN), Complete Graph Network (CGN), Triangular Lattice Network (TLN), and Square Lattice Network (SLN). The section also comprises the comparison with the method without MPEP and the corresponding relative gains. In section \ref{sec5}, we briefly describe the algorithms for the irregular topology, which is the Erd\H{o}s-R\'enyi network (ERN), and present the results accordingly. Finally, in section \ref{sec6}, we conclude the article.   

\section{Preliminaries}
\label{sec2}
In this section, we briefly describe the related concepts relevant to this article. We first provide the quantum network model considered in our work. Then we present the entanglement distribution scheme used for connecting two network nodes, which establishes an end-to-end private entangled channel between the two nodes. This channel can be used for teleportation, for example. We end this section by describing the average teleportation fidelity of the network which is used as a metric to assess the quality of the quantum network in the rest of our paper.
\subsection{Quantum Network Model}
Let us consider $G_Q(N,L)$ to be a quantum network (QN), where $N$ is the number of vertices and $L$ is the number of edges representing the physical channels present in the network which are used to generate elementary entangled links between the adjacent nodes. 
In our model each link is represented by a bipartite isotropic state characterized by its visibility $p$,  and given by,
\begin{eqnarray}
    \rho(p) &=& \frac{1-p}{4}\mathbb{I}+p\dyad{\Phi^+}.
    \label{eq:isotropic_state}
\end{eqnarray}
The fidelity of the above entangled state $\rho(p)$ with respect to the Bell state $\ket{\Phi^+}=(\ket{00}+\ket{11})/\sqrt{2}$, can be expressed as $f=\mel{\Phi^+}{\rho}{\Phi^+}=\frac{1}{4}(1+3p)$ with $0\leq f \leq 1$. Importantly, the isotropic state is entangled for $p>\frac{1}{3}$ and $f>\frac{1}{2}$.  

Physically, the isotropic state is an outcome of a depolarizing channel where the probability of becoming a maximally mixed state (completely noisy) is $1-p$ and $p$ is the probability that the state remains unaffected. In a real network based on photonic qubits, the parameter $p$ (as well as $f$) may vary across the edges based on different factors \cite{Ferreira2024}, leading to a distribution of $p \in [0,1]$ over the entire network. In our work for simplicity, we assume a uniform value of $p$ distributed over the network.

\subsection{Entanglement Distribution Scheme in a QN}
In a connected quantum network, any two network nodes are connected via at least one network path consisting of multiple intermediate nodes. In complex quantum networks, there might exist multiple, alternative and distinct (MAD) network paths between two nodes. In such a network, the aim is to distribute an end-to-end entanglement connection between two nodes, which acts as a private channel for communication. The distribution of end-to-end entanglement however, requires utilization of quantum operations such as entanglement swapping (ES) \cite{swapping_sohel} and entanglement purification (EP) \cite{deutsch1996} with the elementary entangled states along the edges of the network paths. There can exist a plethora of entanglement distribution schemes with different ways of implementation of the two network operations \cite{briegel1998quantum, pant_routing, victora_routing,alejandra_routing,abane_routing}. However, we choose the multipath entanglement purification (MPEP) \cite{purif_bala} for entanglement distribution in our model.

ES is an LOCC protocol that establishes entanglement between two parties which have never interacted with each other. This is a necessary operation to increase the range of entanglement in a QN. Along a network path between two intended nodes $S$ and $T$, ES is performed at the intermediate nodes, which distributes an entangled connection between $S$ and $T$, see Fig. \ref{fig:ES}. Consider a network path of length $\ell$ and the entangled state at $i$th edge of the path is an isotropic state characterized by visibility $p_i$ and fidelity $f_i$. Then the end-to-end state distributed between $S$ and $T$ after swapping at the intermediate nodes will be of the form of Eq.~(\ref{eq:isotropic_state}) and its visibility and fidelity can be expressed as,
\begin{align}
    p_\text{ES}^{(\ell)}=\prod_{i=1}^\ell p_i,~ f_\text{ES}^{(\ell)}=\frac{1}{4}+\frac{3}{4}\prod_{i=1}^\ell \frac{4f_i-1}{3}.
\end{align}
Note that the visibility and the fidelity of the post swapped end-to-end entangled state are lower than those of the initial states along the edges of the network path. In other words, ES degrades the quality of the distributed entangled state. In the MPEP scheme, EP is therefore performed with multiple end-to-end entangled states between $S$ and $T$ distributed via MAD paths to enhance the quality of the distributed entangled state, see Fig. \ref{fig:CG_schm}. The MPEP protocol is accomplished via the following steps:
\begin{enumerate}
    \item Find $z$ multiple, alternative and distinct network paths $\{\mathcal{R}_1,\mathcal{R}_2,\dots,\mathcal{R}_z\}$ between $S$ and $T$.
    \item Establish $z$ end-to-end states $\{\rho_1^{ST},\rho_2^{ST},\dots,\rho_z^{ST}\}$ between $S$ and $T$ via swapping at the intermediate nodes along each path.
    \item Perform sequential EP with the states distributed in step 2 to obtain a high quality end-to-end state after $(z-1)$ rounds of purifications.
\end{enumerate}

\begin{figure}
    \centering
    \includegraphics[scale=0.18]{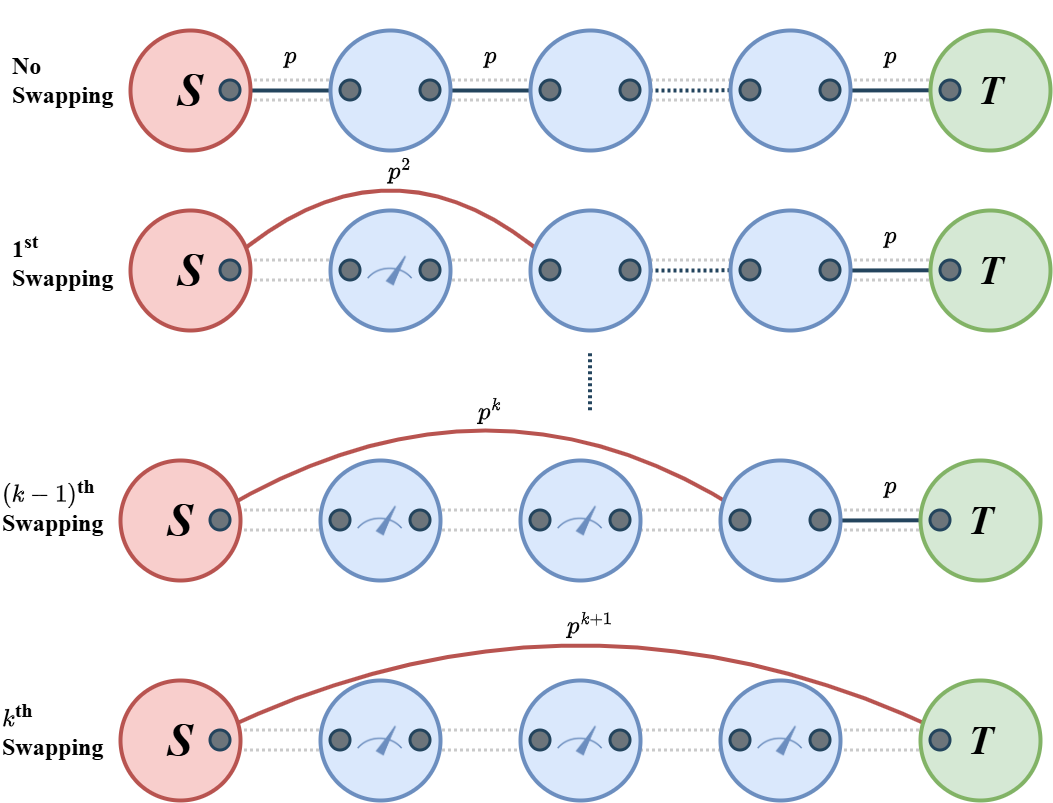}
    \caption{Step-by-step Entangled Swapping scenario: The source ($S$) and the target ($T$) are not initially connected directly. However, there are intermediate repeater stations in between them. After $k$\textsuperscript{th} swapping, there is a connection between $S$ and $T$ via an isotropic state with parameter $p^{k+1}$.}
    \label{fig:ES}
\end{figure}

We now describe the EP protocol proposed by Deutsch et. al. \cite{deutsch1996} used for sequential purification of the entangled states distributed along multiple paths in the MPEP scheme discussed above. In this purification protocol, two Bell diagonal states with fidelities $f_1,f_2>\frac{1}{2}$ are considered to be shared between two parties, say Alice and Bob. They apply bilateral CNOT gates locally by choosing one entangled pair as source and the other pair as target. Then they measure the target qubits in the computational basis and communicate the measurement outcomes via a classical channel. If the outcomes of Alice and Bob match, they keep the source pair, which will be a Bell diagonal state with average fidelity expressed as,
\begin{align}
    f_\text{EP}=\mathcal{P}_D(f_1,f_2)=\frac{10f_1f_2-f_1-f_2+1}{8f_1f_2-2f_1-2f_2+5},
\end{align}
where $\mathcal{P}_D$ is the function representing the purification of two Bell diagonal states of fidelities in the arguments. The purification operation has the following properties \cite{deutsch1996, purif_dur_review, mondal2024},
\begin{itemize}
    \item $\mathcal{P}_D(f_1,f_2)$ is strictly monotonically increasing with respect to both arguments for fidelities in the domain $(\frac{1}{2},1)$.
    \item $\mathcal{P}_D(f_1,f_2)$ is always greater than the minimum of $f_1$ and $f_2$, i.e., $\mathcal{P}_D(f_1,f_2)>\min(f_1,f_2)$ for $f_1,f_2\in(\frac{1}{2},1)$.
    \item Purification is commutative, i.e., $\mathcal{P}_D(f_1,f_2)=\mathcal{P}_D(f_2,f_1)$
    \item Purification is not associative, i.e., $\mathcal{P}_D(\mathcal{P}_D(f_1,f_2),f_3)\neq\mathcal{P}_D(f_1,\mathcal{P}_D(f_2,f_3))$.
\end{itemize}
The last property of the purification protocol is of particular interest to us. The non-associativity indicates that the output fidelity after sequential purification of multiple entangled states depends on the order of the fidelities of the states in which they are purified.

\subsection{Teleportation fidelity and Average Teleportation fidelity in a network}
Teleportation is the well-celebrated quantum operation for transferring quantum information from one location to another without directly sending the information qubit with the help of a shared entangled state. However, teleportation is not only restricted to quantum information but can also be extended to the process of transfer of quantum coherence  \cite{sohail2023teleportation}. If the Bell state \cite{bennett1993teleporting} is shared between the sender and receiver, perfect teleportation of the quantum state is possible. However, a shared state is called to be a resource for teleportation if one can teleport the unknown qubit with a fidelity better than what can be best achieved classically \cite{horodecki1996teleportation}. For a general two-qubit mixed state,
\begin{equation}
 \rho= \frac{1}{4}(\mathbb{I}\otimes \mathbb{I} +\sum_{i=1}^3 r_i \sigma_i \otimes \mathbb{I} + \sum_{j=1}^3 s_j \mathbb{I} \otimes \sigma_j +\sum_{i,j} t_{ij} \sigma_i \otimes \sigma_j ) ,  
\end{equation}
the teleportation fidelity is given  by $F_{\max}=\frac{1}{2}(\frac{1}{3} + N(\rho))$, where $N(\rho)= \Tr(\sqrt{T^{\dagger}T})$, with $T= \Tr(\rho(\sigma_i \otimes \sigma_j))$ being the correlation matrix and $\sigma_i$ being the Pauli matrices. Here, $r_i$ and $s_j$ are the local Bloch vectors. It is interesting to note that classically, without sharing an entangled state, the receiver can estimate the entanglement of the unknown qubit on the source's side by integrating over the Bloch sphere with a fidelity $\frac{2}{3}$ \cite{horodecki1996teleportation}. Any shared entangled state that helps us to obtain a fidelity greater than $\frac{2}{3}$ will be called as a resource and shows quantum advantage. It is important to note that not all entangled states are useful for teleportation, as there will be an entangled state for which the teleportation fidelity is below $\frac{2}{3}$ \cite{horodecki1996teleportation,chakrabarty2010teleportation,chakrabarty2011deletion,adhikari2008quantum}. As a result, not all entangled witness operators will act as a teleportation witness \cite{ganguly2011entanglement}. For the isotropic state given in Eq. \eqref{eq:isotropic_state}, the state is useful for teleportation whenever $p>\frac{1}{3}$, which is the same range for the state to be entangled. \\
\indent If we consider a repeater based quantum network with adjacent elementary links already established, and then consider teleportation between any source $S$ node and a target $T$, we need to do the entanglement swapping if they are not connected, as mentioned in the previous subsection.  This will eventually give rise to a desired link between any source and a target. Now, if we want to calculate the average fidelity of teleportation of the entire network as defined in \cite{mylavarapu2024teleportation}, we have to take the average over all possible pairs of source and target. If $S$ and $T$ are connected via multiple paths, let $\mathcal P_{\rm max}$ be the path with the maximum fidelity. We get the average of the highest-achievable teleportation fidelity if we take the average of $F^{\rm max}_{{\rm ST},\mathcal P_{\rm max}}(\rho)$ over all possible combinations of $S$ and $T$ (i.e., any pair of nodes can be the source and the target):
\begin{equation}
F^{\rm max}_{\rm avg}(\rho) = \langle F^{\rm max}_{{\rm ST},\mathcal P_{\rm max}}(\rho)\rangle_{{\rm ST}}= \langle F^{\rm max}_{\mathcal P}(\rho)\rangle_{\mathcal P}.   \end{equation}
Since different topologies have different path lengths, the expression for the average of maximum teleportation fidelity will be different for different topologies. When the intermediate links are not isotropic states but some other noise state, the efficient protocol proposed in \cite{ghosal2025repeater} can be implemented to ensure the optimal teleportation fidelity.

\section{Multipath Entanglement Purification in a Quantum Network}
\label{sec3}
In this section, we propose two different algorithmic frameworks when edges are available for one time and multiple times for entanglement purification, followed by the identification of an optimal purification strategy out of these two methods for a quantum network. 

\subsection{Strategies for MPEP and Calculation of Teleportation Fidelity}
Here we describe how to implement a multipath purification strategy based on \cite{mondal2024}. We consider a homogeneous QN in which entangled states are distributed along $r_{\max}$ distinct paths\footnote{We denote the path of length $r$ by $l_r$. The terms ``path $l_r$" and ``path of length $l_r$" are equivalent, and we use these terms interchangeably throughout the article.} between any source-target pair $(S,T)$. In this scenario, we can acquire different purification methods; one such method is denoted by $\mathcal{M}(l_{r_1}, \hdots l_{r_{\max}})$. In this method, the states are distributed along the first two paths
are purified initially. The resultant purified state will be successively purified with the state distributed along the next path. 
Specifically, at each step, the output of the previous purification is combined with the state corresponding to the next path in the sequence. This iterative process continues till the final path 
is incorporated, requiring a total of $r_{\max}-1$ purification steps. Schematically,

\begin{eqnarray} 
\mathcal{M}(l_{r_1}, \hdots, l_{r_{\max}}) :=\underbrace{\underbrace{\underbrace{\underbrace{l_{r_1} \oplus l_{r_2}}_{\text{1\textsuperscript{st} Purification}} \oplus l_{r_3}}_{\text{2\textsuperscript{nd} Purification}} \oplus \hdots \oplus l_{r_i}}_{\text{$(i-1)$\textsuperscript{th} Purification}} \oplus \hdots
 \oplus l_{r_{\max}}}_{\text{$(r_{\max}-1)$\textsuperscript{th} Purification}}. \nonumber
\end{eqnarray} 
%
The teleportation fidelity can then be evaluated using Eq. \eqref{eq:A3} in the appendix \ref{appA}, once the final state corresponding to the $(S,T)$ is obtained via the purification method $\mathcal{M}$. We denote this as
\begin{eqnarray}
    F^{\mathrm{tel}}_{S,T}(\mathcal{M}(l_{r_1}, \hdots, l_{r_{\max}})) := F^{\mathrm{tel}}_{S,T}(l_{r_1}, \hdots, l_{r_{\max}}).
\end{eqnarray}
For example, if there are paths of lengths $l_1,l_2$ and $l_3$ between a $(S,T)$ pair, then $F^{\mathrm{tel}}$ of the final state after purification is denoted as $F^{\mathrm{tel}}(l_1, l_3, l_2)$ for a method $\mathcal{M}(l_1, l_3,l_2)$. In this article, we consider two different purification strategies.\\ \\
\noindent \textbf{Shortest path first (SPF) strategy:} In this strategy, we identify the states associated with the shorter paths initially, which suggests the states related to the shorter paths will be purified at first. Eventually, we iteratively purify the rest up to the final. \\
\indent For the mentioned example, the teleportation fidelity of the state obtained at the end of method $\mathcal{M}(l_1, l_2, l_3)=l_1 \oplus l_2\oplus l_3$ is given by,
\begin{eqnarray}
    &&F^{\mathrm{tel}}(l_1, l_2, l_3) \nonumber\\ 
    &=& \frac{3+6p^3+2p^4+2p^5+11p^6}{6+6p^3+12p^6} \nonumber\\
    &=& \frac{11-5f-5f^2+12f^3-10f^4-10f^5+88f^6}{3(7-f-f^2+6f^3-8f^4-8f^5+32f^6}.
\end{eqnarray}
This formula can be iteratively extended by integrating additional paths.\\

\noindent \textbf{Shortest path last (SPL) strategy:} In the case of this strategy, we first identify all possible paths and the states distributed along them. Subsequently, we purify the states associated with the longer paths at first, and iteratively purify the next up to the final. \\
\indent For the mentioned example, the teleportation fidelity of the state obtained at the end of method $\mathcal{M}(l_1, l_2, l_3)=l_3\oplus l_2\oplus l_1$ is given by,
\begin{eqnarray}
    && F^{\mathrm{tel}}(l_3, l_2, l_1) \nonumber \\ 
    &=& \frac{3+p+2p^3+2p^4+5p^5+11p^6}{6+6p^5+12p^6} \nonumber\\
    &=& \frac{11+4f-5f^2-15f^3-10f^4+8f^5+88f^6}{3(7+2f-f^2-9f^3-8f^4+4f^5+32f^6}.
\end{eqnarray}
By incorporating more paths, this formula can be expanded iteratively.
\subsection{Proposed Algorithm for calculation of the Average of Teleportation Fidelity}
We now examine whether it is possible to implement the multipath purification strategy for a homogeneous quantum network, and to calculate $\Ftel$.
In this view, we propose two algorithms based on the total number of uses of one edge. The use of each edge $k$ times indicates that, for a fixed source-target pair $(S,T)$ and while considering all paths $\gamma \in \mathcal{P}_{S,T}$, no edge is utilized by different paths at most than $k$ times.\\
 %
\indent Algorithm \ref{alg:fidelity1} corresponds to the scenario in which each edge in the network can be used at most once.
We portray this algorithm in Fig. \ref{fig:CG_schm}(a), where we consider a complete graph network with 4 nodes. For a given $(S,T)=(1,3)$, we first identified $z$ paths based on the fact that one edge can be used only once. Here, we choose $z=3$, which is the maximum number of paths. Then, after the entanglement swapping process on each path, we end up with the end-to-end states associated with each path. Finally, the purification process is performed to reach the final state, after which we can calculate the teleportation fidelity. Finally, the $\Ftel$ can be calculated by accounting for all possible source and target pairs, which is equal to $F^{\text{tel}}_{S,T}$ for the complete graph.\\
\indent  Algorithm \ref{alg:fidelity2} allows all edges to be used multiple times. This algorithm is depicted in Fig. \ref{fig:CG_schm}(b), where we consider the same complete graph network with 4 nodes. For a given $(S,T)=(1,3)$, we first identified $z$ paths based on the fact that one edge can be used a maximum number of times, while other steps remain the same as the previous one. Here we take $z=5$, which is the maximum number of paths that can be found. Finally, the $\Ftel$ is calculated by accounting for all possible source and target pairs, which is again equal to $F^{\text{tel}}_{S,T}$ for the complete graph.

\begin{tcolorbox}[myalgobox, title=Algorithm 1: Calculation of \ensuremath{F_{\mathrm{avg}}^{\mathrm{tel}}} via SPL (SPF) purification Strategy (\ensuremath{k = 1})]
\begin{algorithmic}[1]

\State Identify all source and target node pairs $(S,T)\in V$
\label{alg:fidelity1}
\ForAll{$(S,T) \in V$}
    \State Find a total of $z$ number of paths $\mathcal{P}_{S,T}$ between $S$ and $T$ where one edge can be used only once ($k$=1)
    \State Sort $\mathcal{P}_{S,T}$ in descending (ascending) order of path length
    
    \ForAll{path $\gamma \in \mathcal{P}_{S,T}$}
        \State Perform entanglement swapping along path $p$
        \State Compute the resulting final quantum state
    \EndFor
    \State Perform entanglement purification between the two states corresponding to the longest (shortest) path lengths
    \State Compute the teleportation fidelity $F^{\text{tel}}_{S,T}$
\EndFor
\State Compute the average of maximum teleportation fidelity:
$$
F_{\mathrm{avg}}^{\mathrm{tel}} = \frac{1}{|V|} \sum_{\substack{S, T \in V \\ S \neq T}} F^{\text{tel}}_{S,T}
$$
\end{algorithmic}
\end{tcolorbox}
\begin{tcolorbox}[myalgobox, title=\textbf{Algorithm 2}: Calculation of \ensuremath{F_{\mathrm{avg}}^{\mathrm{tel}}} via SPL (SPF) purification Strategy (\ensuremath{k = k_{\mathrm{max}}})]

\begin{algorithmic}[1]
\State Identify all source and target node pairs $(S,T)\in V$
\ForAll{$(S,T) \in V$}\label{alg:fidelity2}
    \State Find a total of $z$ number of paths $\mathcal{P}_{S,T}$ between $S$ and $T$ where one edge can be used maximum times ($k = k_{\mathrm{max}}$)
    \State Sort $\mathcal{P}_{S,T}$ in descending (ascending) order of path length
    
    \ForAll{path $\gamma \in \mathcal{P}_{S,T}$}
        \State Perform entanglement swapping along path $p$
        \State Compute the resulting final quantum state
    \EndFor
    
    \State Perform entanglement purification between the two states corresponding to the longest (shortest) path lengths
    \State Compute the teleportation fidelity $F^{\text{tel}}_{S,T}$
\EndFor

\State Compute the average of maximum teleportation fidelity:
$$
F_{\mathrm{avg}}^{\mathrm{tel}} = \frac{1}{|V|} \sum_{\substack{S, T \in V \\ S \neq T}} F^{\text{tel}}_{S,T}
$$
\end{algorithmic}

\end{tcolorbox}
\begin{figure*}        
    \centering
    \includegraphics[width=0.85\linewidth]{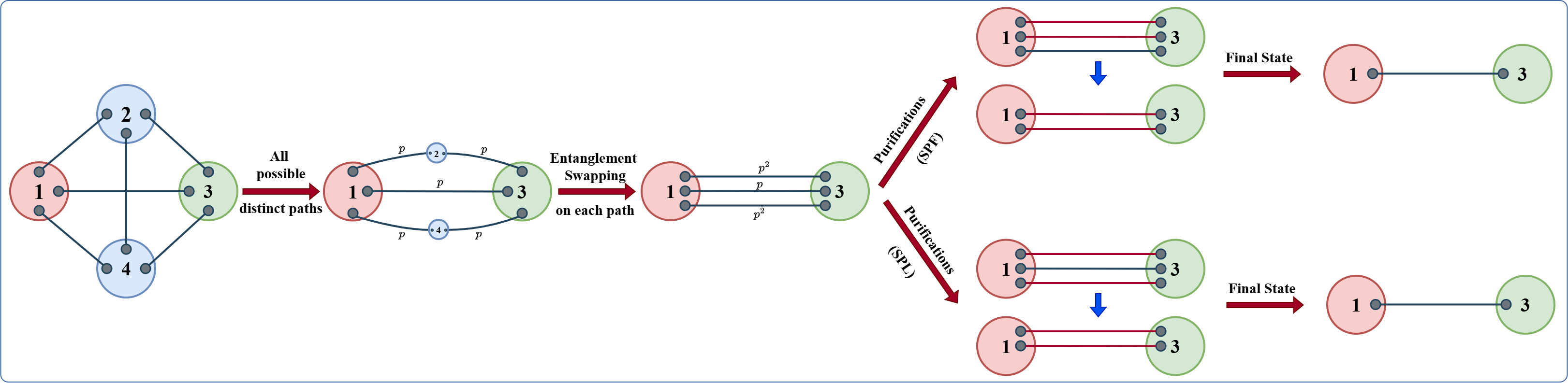}
    \\[1ex]
    (a) For $k = 1$.
    \\[1ex]
    \includegraphics[width=0.85\linewidth]{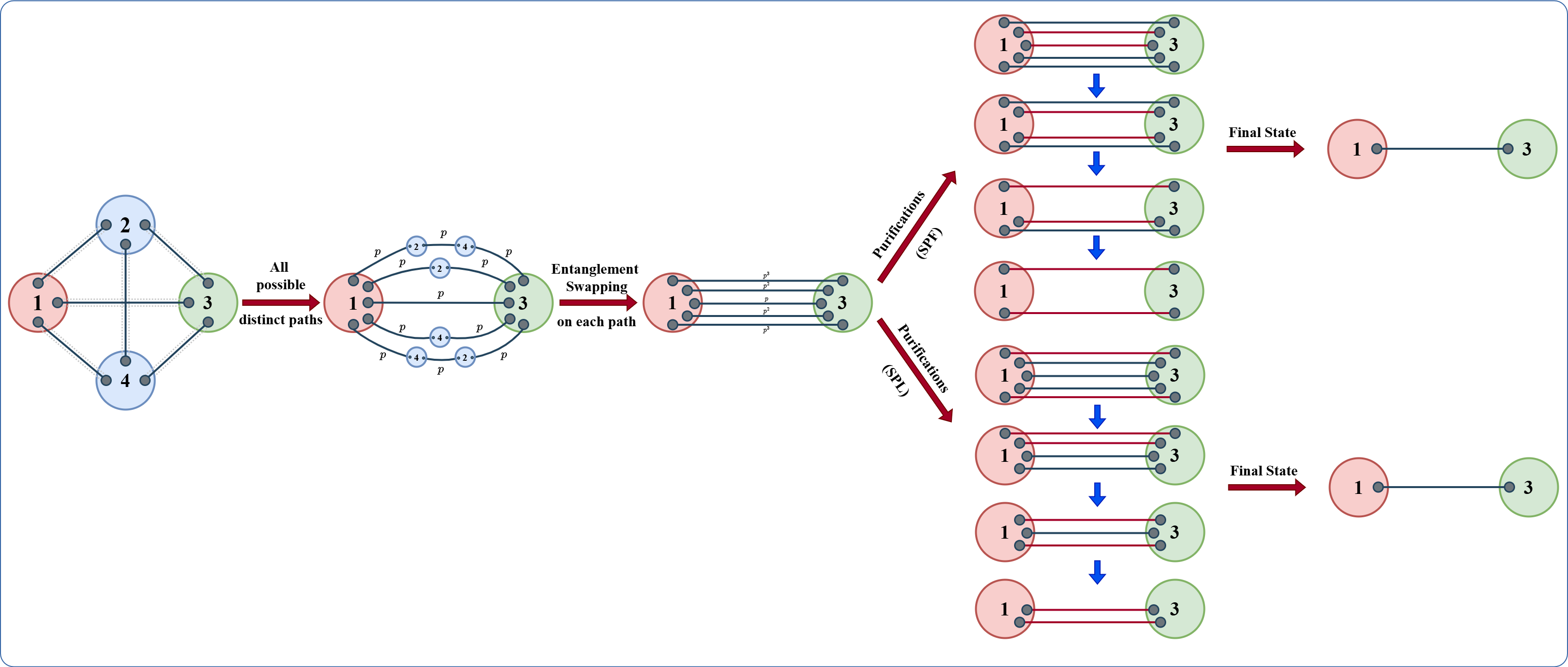}
    \\[1ex]
    (b)  For $k = k_{\max}$.
    \caption{Algorithm for evaluating $\Ftel$ using the SPF and SPL purification strategies ($k=1$ and $k=k_{\max}$) for a complete graph with $N=4$. The procedure is illustrated for a single source–target pair, $(S,T)=(1,3)$. The remaining pairs can be treated analogously, and the $\Ftel$ is obtained by averaging over all pairs.}
    \label{fig:CG_schm}
\end{figure*}

%


    
    



\subsection{Comparison of SPF and SPL MPEP Strategies} 
In this section, we provide the superior strategy among SPF and SPL for entanglement purification in the quantum network. We first provide the superior sequence for purification of $n$ quantum states with fidelities $f_1>f_2>\dots>f_n$ in the following Lemma \ref{lemma_purification}. Then we provide the superior entanglement purification strategy in the Theorem \ref{thm1}. Now let us introduce a notation for denoting the output fidelity with a given sequence of purification as, $\mathcal{L}(x_1,x_2,\dots,x_n)=\mathcal{P}_D(\dots\mathcal{P}_D(\mathcal{P}_D(x_1, x_2), x_3), \dots, x_n)$, where states with fidelities $x_1$ and $x_2$ are used in first purification step, then $x_3$ is purified with the output of the first purification, and so on upto $x_n$.
\begin{lemma}
    In a sequential purification process with more that two Bell diagonal states, the sequence with increasing order of fidelities will lead to higher output fidelity than the sequence with decreasing order of fidelities. Mathematically,
\begin{align}
    \mathcal{L}(f_n,f_{n-1},\dots,f_1)>\mathcal{L}(f_1,f_2,\dots,f_n)
\end{align}
\label{lemma_purification}
\end{lemma}
\begin{proof} 
To prove the above lemma we first prove the inequality for three fidelities ($n=3$). Consider three states with fidelities $f_1,f_2,f_3\in(\frac{1}{2},1)$ such that $f_1>f_2>f_3$. If we perform sequential purification with these states, then we want to prove,
\begin{align}
    \mathcal{L}(f_3,f_2,f_1)>\mathcal{L}(f_1,f_2,f_3).
    \label{ineq_pur}
\end{align}
The output fidelities after performing sequential purification in the two different orders of the fidelities can be obtained as,
\begin{align}
    F^\uparrow &= \mathcal{L}(f_3,f_2,f_1)=\frac{N^\uparrow}{D^\uparrow}\nonumber\\
    F^\downarrow &= \mathcal{L}(f_1,f_2,f_3)=\frac{N^\downarrow}{D^\downarrow},
\end{align}
where the quantities in the numerator and the denominator of the two expressions are given by,
\begin{align}
N^{\uparrow} = 92f_1f_2f_3 &- 2f_2f_3 - 8f_1f_2 \nonumber\\&- f_2 - 8f_1f_3 - f_3 + 5f_1 + 4, \nonumber \\
D^{\uparrow} = 64f_1f_2f_3 &+ 20f_2f_3 - 4f_1f_2 \nonumber\\&- 8f_2 - 4f_1f_3 - 8f_3 - 2f_1 + 23, \nonumber\\
N^{\downarrow} = 92f_1f_2f_3&-2f_1f_2-8f_1f_3-f_1\nonumber\\&-8f_2f_3-f_2+5f_3+4, \nonumber\\
D^{\downarrow} = 64f_1f_2f_3 &+ 20f_1f_2 - 4f_1f_3\nonumber\\& - 8f_1 - 4f_2f_3 - 8f_2 - 2f_3 + 23.
\end{align}
The difference of the two numerators turns out to be $N^\uparrow-N^\downarrow=6(f_1-f_3)(1-f_2)$. Since $f_1>f_3$ and $f_2\in (\frac{1}{2},1)$, we have $(f_1-f_3)>0$ and $(1-f_2)>0$ and hence $N^\uparrow>N^\downarrow$. Similarly, $D^\uparrow-D^\downarrow=6(1-4f_2)(f_1-f_3)$, where $(1-4f_2)<0$ and $(f_1-f_3)>0$ and hence, $D^\uparrow<D^\downarrow$. The two inequalities $N^\uparrow>N^\downarrow$ and $D^\uparrow<D^\downarrow$ together prove that,
\begin{align}
    F^\uparrow = \frac{N^\uparrow}{D^\uparrow}>\frac{N^\downarrow}{D^\downarrow} = F^\downarrow.
\end{align}
This completes the proof of Ineq.~(\ref{ineq_pur}). We further need to show a property of the purification sequence to prove the lemma. 

Consider the fidelities $x_1,x_2,\dots,x_n\in(\frac{1}{2},1)$ and another two fidelities $a,b\in(\frac{1}{2},1)$ such that $a>b$. Then it holds that,
\begin{align}
    \mathcal{L}(a,x_1,\dots,x_n)>\mathcal{L}(b,x_1,\dots,x_n).
    \label{ineq_2}
\end{align}
We prove the above relation by mathematical induction. First note that,
\begin{align}
\mathcal{L}(a,x_1,\dots,x_n)=\mathcal{L}(\mathcal{P}(a,x_1),x_2,\dots,x_n).
\end{align}
For $n=1$, the result holds since $\mathcal{P}(a,x_1)>\mathcal{P}(b,x_1)$, which follows from the monotonicity of purification. Now assume that the statement is true for $n$. For $n+1$, we write,
\begin{align}
\mathcal{L}(a,x_1,\dots,x_{n+1})=\mathcal{L}(\mathcal{L}(a,x_1,\dots,x_n),x_{n+1}).
\end{align}
From the induction hypothesis,
\begin{align}
\mathcal{L}(a,x_1,\dots,x_n)>\mathcal{L}(b,x_1,\dots,x_n).
\end{align}
Since for two arguments $\mathcal{L}(A,x)=\mathcal{P}(A,x)$, and $\mathcal{P}$ is strictly increasing in its first argument, it follows that,
\begin{align}
\mathcal{L}(\mathcal{L}(a,x_1,\dots,x_n),x_{n+1})>\mathcal{L}(\mathcal{L}(b,x_1,\dots,x_n),x_{n+1}).
\end{align}
This proves the result for $n+1$, and hence the statement holds for all $n$. Using the Ineq.~(\ref{ineq_pur}) and the property provided by Ineq.~(\ref{ineq_2}), we now provide the proof of the lemma.

Let us introduce the notation, $X_k=\mathcal{L}(f_k,f_{k-1},\dots,f_1)$ and note that, $X_{k+1}=\mathcal{L}(f_{k+1},X_{k})$. Now consider the quantity $\mathcal{L}(f_1,f_2,f_3,\dots,f_n)=\mathcal{L}(\mathcal{L}(f_1,f_2,f_3),\dots,f_n)$ of the fidelities $f_1>f_2>f_3>\dots>f_n$. Using Ineq.~(\ref{ineq_pur}) we get $\mathcal{L}(f_1,f_2,f_3)<\mathcal{L}(f_3,f_2,f_1)$. Now we utilize Ineq.~(\ref{ineq_2}) to obtain,
\begin{align}
    \mathcal{L}(f_1,f_2,f_3,\dots,f_n)<\mathcal{L}(f_3,f_2,f_1,\dots,f_n).
\end{align}
Note that, $\mathcal{L}(f_3,f_2,f_1,\dots,f_n)=\mathcal{L}(X_3,f_4,f_5,\dots,f_n)=\mathcal{L}(\mathcal{L}(X_3,f_4,f_5),\dots,f_n)$ where $X_3>f_4>f_5$ as $X_3>f_3$ due to the second property of purification given in Sec.~\ref{sec2}B. Therefore we can again use Ineq.~(\ref{ineq_pur}) to obtain,
\begin{align}
    \mathcal{L}(X_3,f_4,f_5,\dots,f_n)<\mathcal{L}(f_5,f_4,X_3,\dots,f_n).
\end{align}
Repeating this process, we finally obtain,
\begin{align}
    \mathcal{L}(X_{n-2},f_{n-1},f_n)<\mathcal{L}(f_n,f_{n-1},X_{n-2}),
\end{align}\
which is equivalent to writing,
\begin{align}
    \mathcal{L}(f_n.f_{n-1},\dots,f_1)>\mathcal{L}(f_1,f_2,\dots,f_n).
\end{align}
This completes the proof of Lemma \ref{lemma_purification}.
\end{proof}

Now we provide the superior strategy among LPS and SPS for entanglement purification in a quantum network in the theorem below.
\begin{theorem}
    In a homogeneous quantum network where all the network edges are distributed with identical entangled states, the SPL strategy provides higher output fidelity than the SPF strategy.
    \label{thm1}
\end{theorem}
\begin{proof}
Consider a homogeneous network with all the network edges distributed with isotropic states of isotropic parameter $p$. Also consider that the length of network paths between an arbitrary pair of network nodes is $l_{r_1}<l_{r_2}<\dots<l_{r_n}$. Therefore the end-to-end fidelity of the entangled state distributed along the $i$th path via swapping is, $f_i=\frac{1}{4}(1+3p^{l_{r_i}})$ for $i\in\{1,2,\dots,n\}$. It directly follows that $f_1>f_2>\dots>f_n$. The SPL and the SPF strategies can respectively be expressed as,
\begin{align}
    &\mathcal{M}(l_{r_n},l_{r_{n-1}},\dots,l_{r_1})=\mathcal{L}(f_n,f_{n-1},\dots,f_1)\nonumber\\
    &\mathcal{M}(l_{r_1},l_{r_2},\dots,l_{r_n})=\mathcal{L}(f_1,f_2,\dots,f_n).
\end{align}
According to Lemma \ref{lemma_purification} we have, $\mathcal{L}(f_n,f_{n-1},\dots,f_1)>\mathcal{L}(f_1,f_2,\dots,f_n)$ and therefore we obtain,
\begin{align}
    \mathcal{M}(l_{r_n},l_{r_{n-1}},\dots,l_{r_1})>\mathcal{M}(l_{r_1},l_{r_2},\dots,l_{r_n}).
\end{align}
This proves that the SPL strategy leads to a better output fidelity than the SPF strategy, and hence SPL is the superior strategy for purification.
\end{proof}

\noindent \textbf{Average Teleportation Fidelity:} The average teleportation fidelity is represented by $F_{\mathrm{avg}}^{\mathrm{tel}}$, which is defined as follows: 
\begin{align} 
F_{\mathrm{avg}}^{\mathrm{tel}} &= \frac{1}{|V|} \sum_{\substack{S, T \in V \\ S \neq T}} F^{\text{tel}}_{S,T},
\end{align} 
where $|V|$ is the total number of source–target pairs $(S,T)$ in the QN. The average teleportation fidelity of the QN according to the SPL and SPF strategy is given respectively as,
\begin{align}
    F_{\mathrm{avg}}^{\mathrm{tel}}\big|_{\mathrm{SPL}}&=\frac{1}{|V|}\sum_{\substack{S, T \in V \\ S \neq T}} F^{\mathrm{tel}}_{S,T}(\mathcal{M}(l_{r_n}^{S,T},l_{r_{n-1}}^{S,T},\dots,l_{r_1}^{S,T})),\nonumber\\
    F_{\mathrm{avg}}^{\mathrm{tel}}\big|_{\mathrm{SPF}}&=\frac{1}{|V|}\sum_{\substack{S, T \in V \\ S \neq T}} F^{\mathrm{tel}}_{S,T}(\mathcal{M}(l_{r_1}^{S,T},l_{r_{2}}^{S,T},\dots,l_{r_n}^{S,T})).
\end{align}
Here $l_{r_i}^{S,T}$ denotes the $i$th path between $S$ and $T$. According to Theorem \ref{thm1}, $\mathcal{M}(l_{r_n}^{S,T},l_{r_{n-1}}^{S,T},\dots,l_{r_1}^{S,T})>\mathcal{M}(l_{r_1}^{S,T},l_{r_{2}}^{S,T},\dots,l_{r_n}^{S,T})$ for each pair $(S,T)$ in the QN. Therefore, we can conclude that,
\begin{eqnarray}
    F_{\mathrm{avg}}^{\mathrm{tel}}\big|_{\mathrm{SPL}} \geq F_{\mathrm{avg}}^{\mathrm{tel}}\big|_{\mathrm{SPF}}.
\end{eqnarray}
Equality is only reached when there are exactly two paths between each $(S,T)$ pair.\\
\indent In the remainder of this paper, we employ the SPL MPEP protocol for entanglement distribution in quantum networks.

\section{Purification with Regular topologies}
\label{sec4}
In this section, we describe the iterative purification formula for regular topologies: the ring network (RN), the complete graph network (CGN), the triangular lattice network (TLN), and the square lattice network (SLN).  For both the situation with MPEP and the case without MPEP, we show how the average teleportation fidelity changes with respect to the isotropic state parameter ($p$), where the number of purifications $i$ is varied. Additionally, we plot the relative gain with respect to $p$, which is defined as follows 
\begin{eqnarray}
\text{Relative Gain}= \frac{\Ftel|_{\text{with MPEP}}-\Ftel|_{\text{without MPEP}}}{\Ftel|_{\text{without MPEP}}}.
\label{rg}
\end{eqnarray}
The same is presented for both the scenarios of $k=1$ and $k=k_{\max}$ to facilitate a comparison of the proposed algorithms. As the SPL strategy outperforms the SPF strategy, the analytical and numerical results presented here are only for the SPL strategy.
%
\subsection{Ring Network (RN)}
As our first demonstration, we consider a ring network (RN) as a loop topology and try to see when there is a multiple path between the source and target, whether the teleportation fidelity of the network is benefited by the purification process by converting multiple paths to a single path.
\begin{figure}[h]
    \centering
    \includegraphics[scale=0.18]{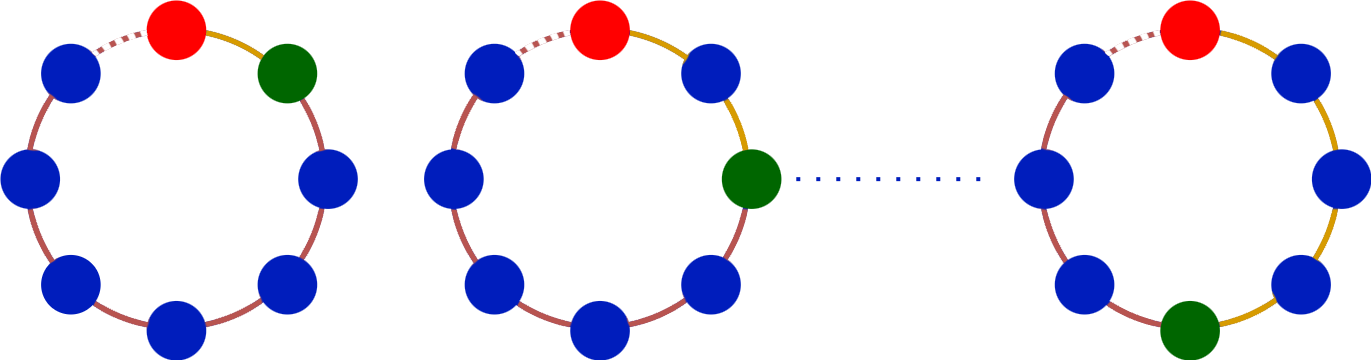}
    \caption{Purification procedure for an $N$-node ring network (RN) for both $k=1$ and $k=k_{\max}$. A fixed source node (marked in red) is considered, and the two distinct paths (marked in different colors) are determined based on the positions of the target nodes (marked in green). Based on the symmetry, the remaining configurations can be obtained analogously after the final step.}
    \label{fig:ringN}
\end{figure}

For more clarification, see the Fig. \ref{fig:ringN}, where we have shown two distinct paths always available for every source and target.
\subsubsection{One edge can be used only once, i.e. $k=1$}

In the case of a ring network with node $N$, there are generally two separate paths from the source ($S$) to the target ($T$), so the maximum number for purification is limited to one.\\
%
\indent We first consider the nearest neighbour case, i.e., the $(S,T)$ is $(j,j+1)$ as shown in Fig. \ref{fig:ringN}. In this situation, there are two distinct paths available, which are $l_1$ and $l_{N-1}$. This scenario occurs $N$ times for a ring network with $N$ nodes. Hence, the corresponding contribution is $$N \times F^{\mathrm{tel}}_{j,j+1}(l_{N-1},l_1),$$ where the argument denotes the ordering of the paths based on the pathlengths. 
Next, consider case i.e. $(S,T) = (j,j+2)$, the available distinct paths are $l_2$  and $l_{N-2}$ and this configuration also occurs $N$ times again which leads to the contribution $$N \times F^{\mathrm{tel}}_{j,j+2}(l_{N-2},l_2).$$ 
Finally, we consider $(S,T)=(j,j+\lfloor\frac{N}{2}\rfloor)$ for which the distinct paths are $l_{\lfloor\frac{N}{2}\rfloor}$ and $l_{N-\lfloor\frac{N}{2}\rfloor}$. The number of such configurations equals $N$ for odd $N$, and for even $N$ it reduces to $N-2$. Hence, collecting all the results and taking the average, we have
\begin{eqnarray}
    \Ftel = 
    \begin{cases}
        \langle N\times F^{\mathrm{tel}}_{j,j+1}(l_{N-1},l_1), N\times F^{\mathrm{tel}}_{j,j+1}(l_{N-2},l_2),\\\hdots,N\times F^{\mathrm{tel}}_{j,j+\lfloor\frac{N}{2}\rfloor}(l_{N-\lfloor\frac{N}{2}\rfloor}, l_{\lfloor\frac{N}{2}\rfloor}) \rangle,\\ \text{when $N$ is odd},\\
        \langle N\times F^{\mathrm{tel}}_{j,j+1}(l_{N-1},l_1), N\times F^{\mathrm{tel}}_{j,j+1}(l_{N-2},l_2),\\\hdots,(N-2)\times F^{\mathrm{tel}}_{j,j+\lfloor\frac{N}{2}\rfloor}(l_{N-\lfloor\frac{N}{2}\rfloor}, l_{\lfloor\frac{N}{2}\rfloor}) \rangle,\\ \text{when $N$ is even}.        
    \end{cases}
\end{eqnarray}
%
        

\subsubsection{One edge can be used maximum number of times i.e. $k=k_{\max}$}
%
\indent We find that for $N=4$, at $p\approx 0.8966$, the value of $\Ftel$ with MPEP strategy exceeds that without MPEP. For $N>4$, such an improvement is not observed for any value of $p$. The $\Ftel$ and relative gain obtained via the MPEP strategy for different values of $p$ are shown in Fig. \ref{Fvsp_ring}(a) and \ref{Fvsp_ring}(b), respectively, for a RN with $N=100$. In the case of RN, only two paths are available between any node pair, thereby restricting the achievable enhancement and relative gain. Moreover, the length of the two available paths differ significantly with increasing number of nodes in the network, which in return decreases the purified fidelity. As a result, no improvement can be seen for $N>4$. This also indicates that the MPEP strategy cannot improve the $\Ftel$ for a ring network.  
%

    
\begin{figure}[h]
    \centering
    \includegraphics[scale=0.27]
        {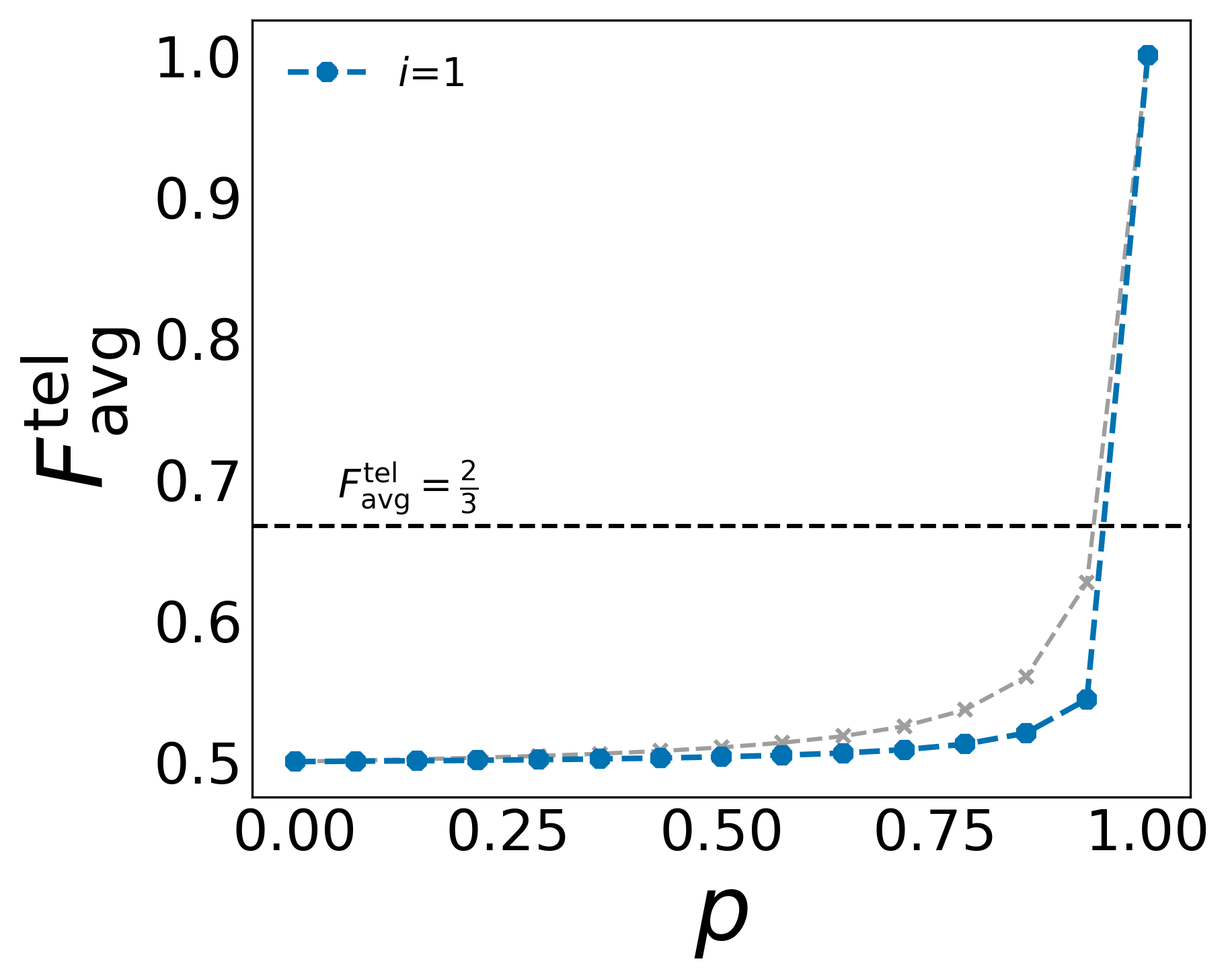}
    \includegraphics[scale=0.27]
        {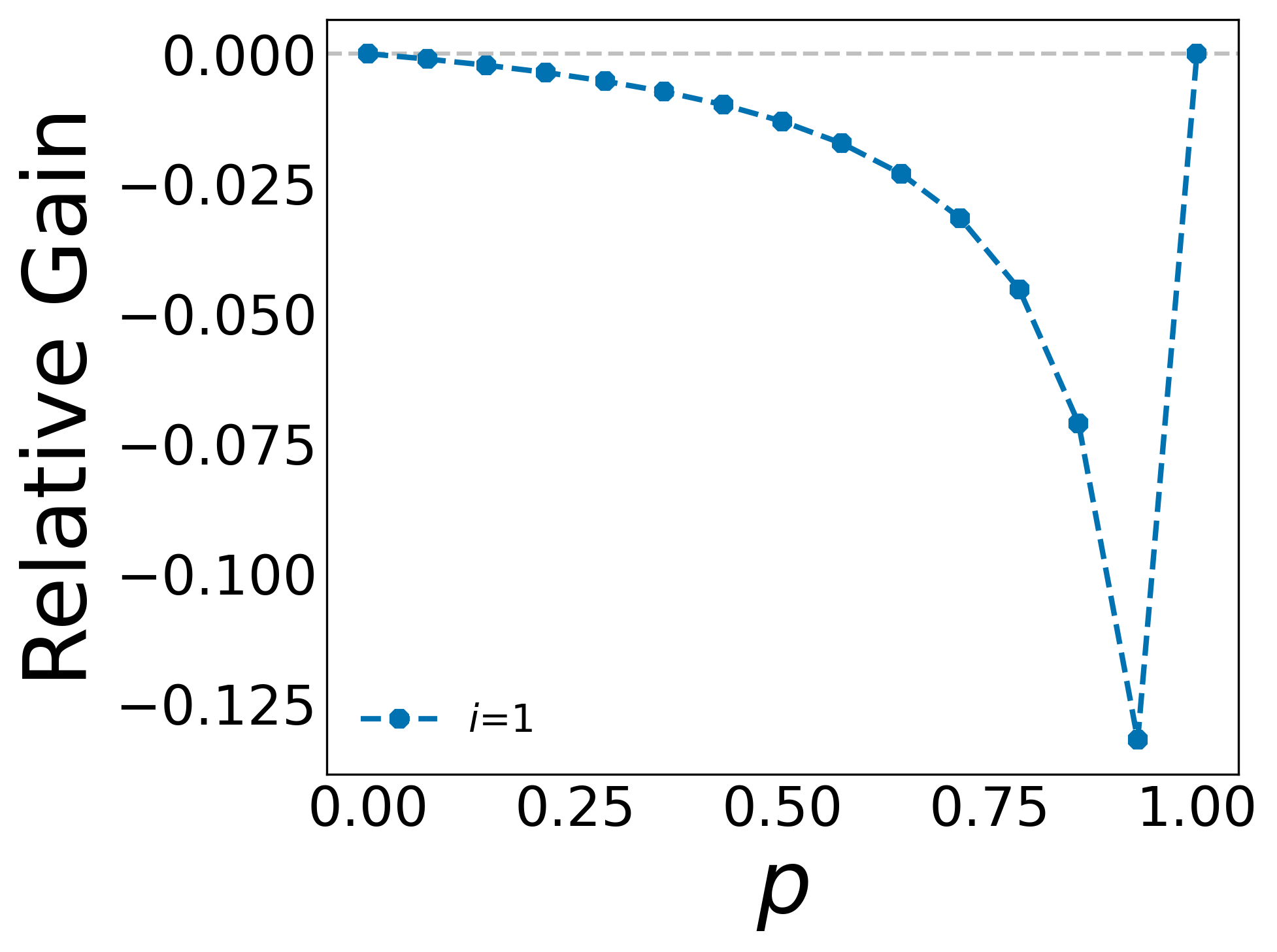}
    \\[1ex]
    \hspace*{0.75cm}(a) \hspace{3.9cm} (b)
    \caption{Plots of (a) $F^{\mathrm{tel}}_{\mathrm{avg}}$ as a function of $p$ and (b) relative gain as a function of $p$ for a ring network with $N=100$ and $k=1=k_{\max}$. The total number of paths is taken to be $z=z_{\max}=2$. The horizontal black line suggests $\Ftel=\frac{2}{3}$. The grey line denotes the variation of $\Ftel$ with respect to $p$ for the without MPEP scenario.}
    \label{Fvsp_ring}
\end{figure}
%
\subsection{Complete Graph Network (CGN)}
Let $k$ denote the number of times an edge can be utilized. As in the previous case, we consider two cases based on how many paths one edge can be part of.

\begin{figure}[h]
    \centering
    \includegraphics[scale=0.135]{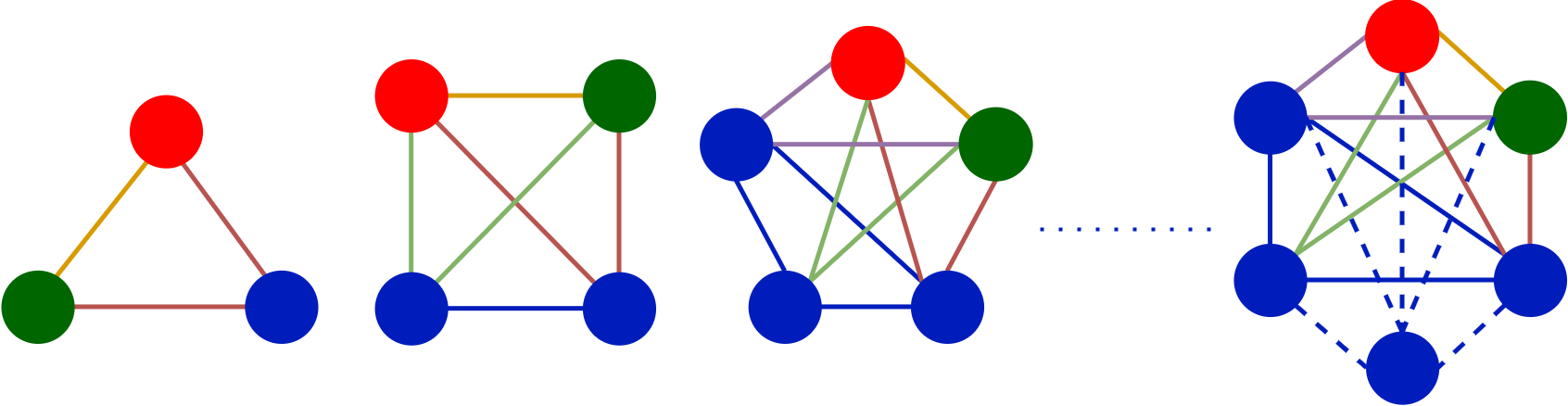}
    \caption{Purification procedure for an $N$-node complete graph network (CGN) for $k=1$. A fixed source (marked in red) and target (marked in green) are considered, and the paths are determined based on the total number of intermediate nodes. As the total number of nodes increases, the number of intermediate nodes—and hence the number of available paths—also increases; the additional paths at each step are indicated by different colors. Due to symmetry, the remaining $(S,T)$ pairs can be treated analogously after the final step.}
    \label{fig:CG_puri}
\end{figure}

\subsubsection{One edge can be used only once, i.e. $k=1$}
Let us consider a fixed source–target pair $(S,T)=  (j,j+1)$. For the base case $N=3$, there are two distinct paths between $S$ and $T$, which are $l_1$ and $l_2$ as depicted in the first subfigure of Fig. \ref{fig:CG_puri}. As the number of nodes increases, each additional node introduces one extra path of length $l_2$ as highlighted by different colors in the subsequent subfigures of Fig. \ref{fig:CG_puri}. Consequently, for a network with $N$ nodes, the total number of $l_2$ paths turns out to be $N-2$. Hence, the set of paths for the iterative process is $$\mathcal{M} = \{l_1, \underbrace{l_2, \hdots, l_2}_{N-2 \text{ times}}\}.$$
Now, let $z$ denote the maximum number of paths used for the purification procedure. Following the SPL strategy, the paths are arranged in descending order. Hence, the contribution corresponding to a single $(S,T)$ pair is $$F^{\mathrm{tel}}_{j,j+1}(\mathcal{M}^{(z)}_{\mathrm{SPL}}) = F^{\mathrm{tel}}_{j,j+1}(\underbrace{l_2, \hdots, l_2}_{z-1 \text{ times}},l_1).$$
From the symmetry of the complete graph topology, the average over all source-target pairs is identical to the contribution obtained for any single pair $(S,T)$. Therefore 
\begin{eqnarray}
    \Ftel = 
        \langle F^{\mathrm{tel}}_{j,j+1}(\mathcal{M}^{(z)}_{\mathrm{SPL}}) \rangle = F^{\mathrm{tel}}_{j,j+1}(\mathcal{M}^{(z)}_{\mathrm{SPL}}).
\end{eqnarray}
%

\subsubsection{One edge can be used maximum number of times i.e. $k=k_{\max}$}
If we consider a specified source-target pair $(S,T) = (j,j+1)$. The shortest path connecting $S$ and $T$ has length $l_1$, while the longest path has length $l_{N-1}$. Now, for $k=k_{\max}$ and for a path of length $l_m$, fixing the endpoints at $j$ and $j+1$, the total number of paths is given by $${}^{N-2}P_{m-1} = \frac{(N-2)!}{(N-m-1)!}.$$ Hence, the set for paths is $$\mathcal{M} = \{l_1, \underbrace{l_2, \hdots, l_2}_{{}^{N-2}P_{1} \text{ times}}, \hdots, \underbrace{l_{N-1}, \hdots, l_{N-1}}_{{}^{N-2}P_{N-2} \text{ times}}\}.$$
Again, if $z$ is the maximum number of paths used for the purification procedure, we can rearrange the paths based on the SPL strategy. Therefore, the contribution for fixed $(S,T)$ is
$$F^{\mathrm{tel}}_{j,j+1}(\mathcal{M}^{(z)}_{\mathrm{SPL}}).$$
The symmetry of the whole network topology indicates that the average across all source-target pairs is equivalent to the contribution derived from any individual pair $(S,T)$. Therefore,
\begin{eqnarray}
    \Ftel &=& 
        \langle F^{\mathrm{tel}}_{j,j+1}(\mathcal{M}^{(z)}_{\mathrm{SPL}}) \rangle = F^{\mathrm{tel}}_{j,j+1}(\mathcal{M}^{(z)}_{\mathrm{SPL}}) .
\end{eqnarray}
%
%
\subsubsection{Results}
In Fig. \ref{Fvsp_CG_k1}(a) we present the scaling of $\Ftel$ with $p$ for both MPEP with varying path lengths and utilizing a single path without using MPEP. For $k = 1$, we can see the variation of $\Ftel$ as we increase the number of purification steps $i$. Here, we find that for $N=100$ and for $z=20$, the value of $\Ftel$ with MPEP strategy exceeds that without MPEP at $p\approx 0.8186$ when $i=5$, and $p \approx 0.8135$ when $i=10$. For the maximum number of purification steps $i=i_{\max} = 19$, the enhancement can be achieved at $p \approx 0.6440$. These results indicate that the MPEP strategy enhances $\Ftel$, with improvements more pronounced as the number of purification steps approaches the maximum. We have also plotted the relative gain of $\Ftel$ with respect to $p$ in Fig. \ref{Fvsp_CG_k1}(b), which also shows a significant positive gain at higher values of $i$.\\
\begin{figure}[h]
    \centering
    \includegraphics[scale=0.27]
        {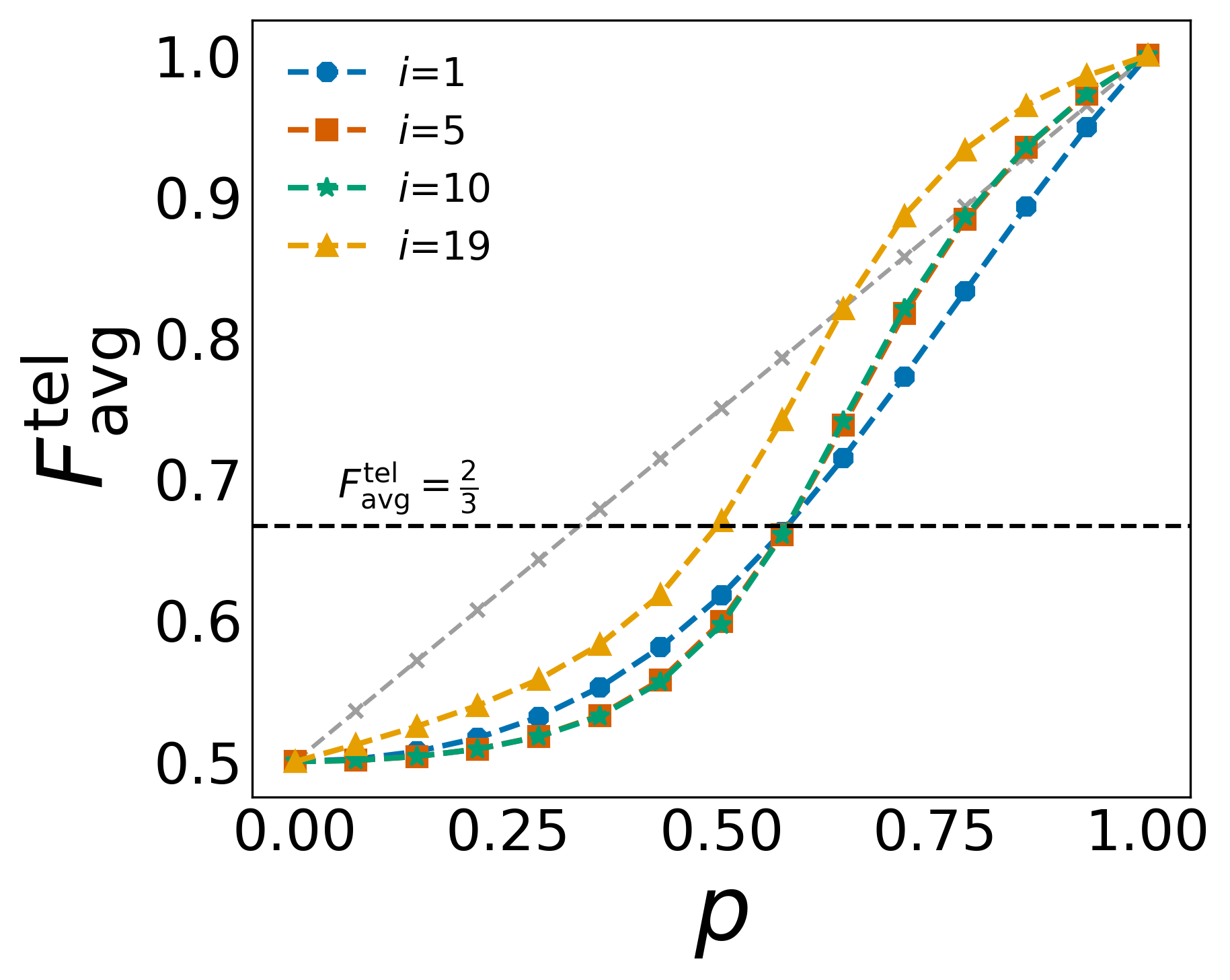}
    \includegraphics[scale=0.27]
        {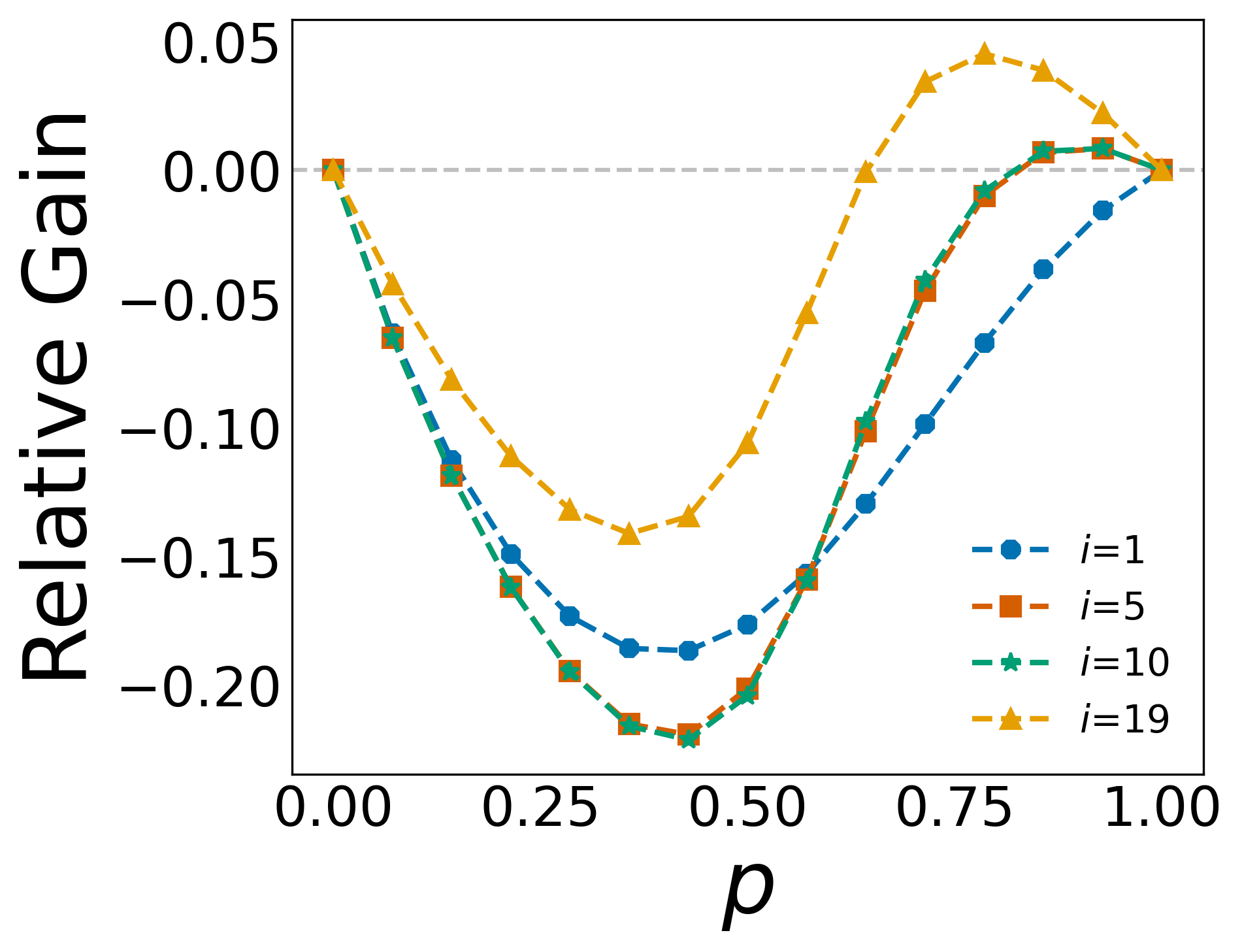}
    \\[1ex]
    \hspace*{0.75cm}(a) \hspace{3.9cm} (b)
    \caption{Plots of (a) $F^{\mathrm{tel}}_{\mathrm{avg}}$ as a function of $p$ and (b) relative gain as a function of $p$ for a complete graph network with $N=100$ and $k=1$. The total number of paths is taken to be $z=20$. The horizontal black line suggests $\Ftel=\frac{2}{3}$. The grey line denotes the variation of $\Ftel$ with respect to $p$ for the without MPEP scenario.}
    \label{Fvsp_CG_k1}
\end{figure}
On the other hand, for $k=k_{\max}$ and $z=20$, the set of paths is the same as the $k=1$ case. Hence, the results remain the same as previous, as shown in Fig. \ref{Fvsp_CG_kmax}(a). The same result is reflected in the case of relative gain as depicted in Fig. \ref{Fvsp_CG_kmax}(b). In this case, the set of paths can contain more elements, and therefore, it is possible to do a larger number of purifications. However, the value of $\Ftel$ will saturate after a certain number of purifications.
\begin{figure}[h]
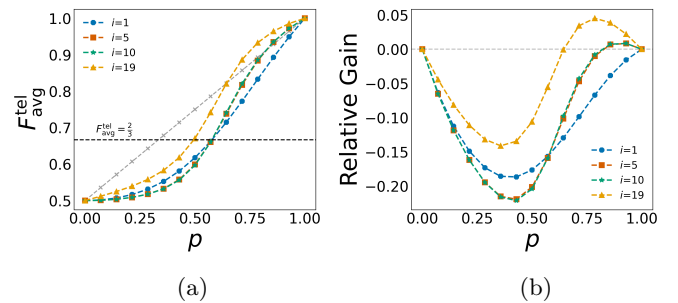

    \centering
    \includegraphics[scale=0.27]
        {Figures/CG/Fvsp_CG_kmax_LPS_N=100.png}
    \includegraphics[scale=0.27]
        {Figures/CG/RGvsp_CG_kmax_LPS_N=100.png}
    \\[1ex]
    \hspace*{0.75cm}(a) \hspace{3.9cm} (b)
    \caption{Plots of (a) $F^{\mathrm{tel}}_{\mathrm{avg}}$ as a function of $p$ and (b) relative gain as a function of $p$ for a complete graph network with $N=100$ and $k=k_{\max}$. The total number of paths is taken to be $z=20$. The horizontal black line suggests $\Ftel=\frac{2}{3}$. The grey line denotes the variation of $\Ftel$ with respect to $p$ for the without MPEP scenario.}
    \label{Fvsp_CG_kmax}
\end{figure}\\
\indent For both $k=1$ and $k=k_{\max}$, the MPEP protocol fails to achieve the quantum advantage threshold, i.e. $\Ftel=\frac{2}{3}$, at a value of $p_c$ lower than that required for without MPEP scenario. Nevertheless, increasing the number of paths $z$ while performing the proposed algorithm does not necessarily lead to a significantly higher value of $\Ftel$. As $z$ increases, the corresponding increment of $\Ftel$ becomes gradually smaller. As a result, $\Ftel$ reaches its saturation value after certain $z$ with $i_{\max}=z-1$. For $N=100$, we found that the increment becomes of the order of $10^{-4}$ after $z=10$ and $i=9=i_{\max}$ for both $k=1$ and $k=k_{\max}$ which is shown in Fig. \ref{Fvsp_CG_saturation}. Therefore, the reuse of the network edges as part of different network paths does not make a difference in this particular network topology, and therefore, one can simply use MPEP with $k=1$.
\begin{figure}[h]
    \centering
    \includegraphics[scale=0.27]{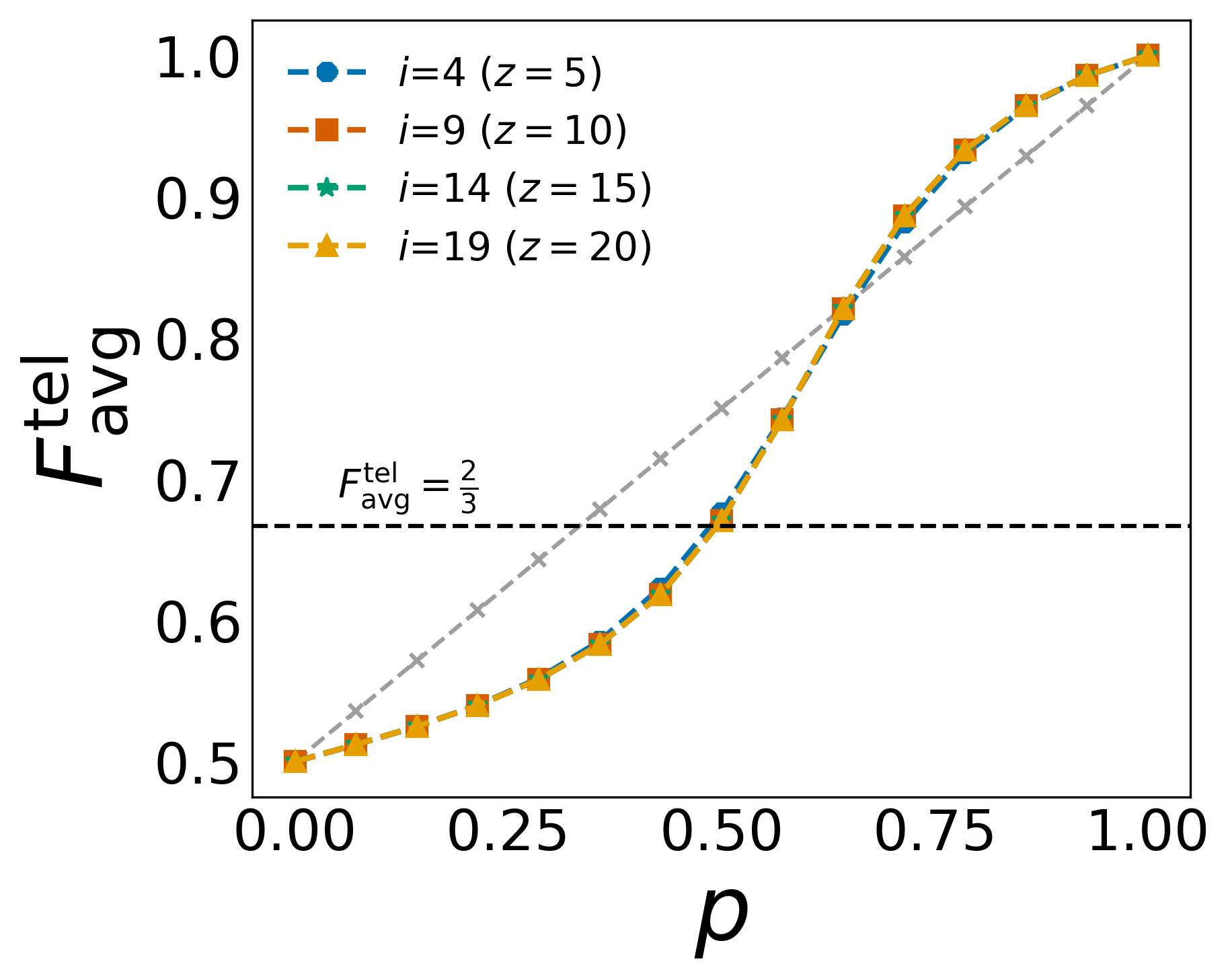}
    \includegraphics[scale=0.27]
        {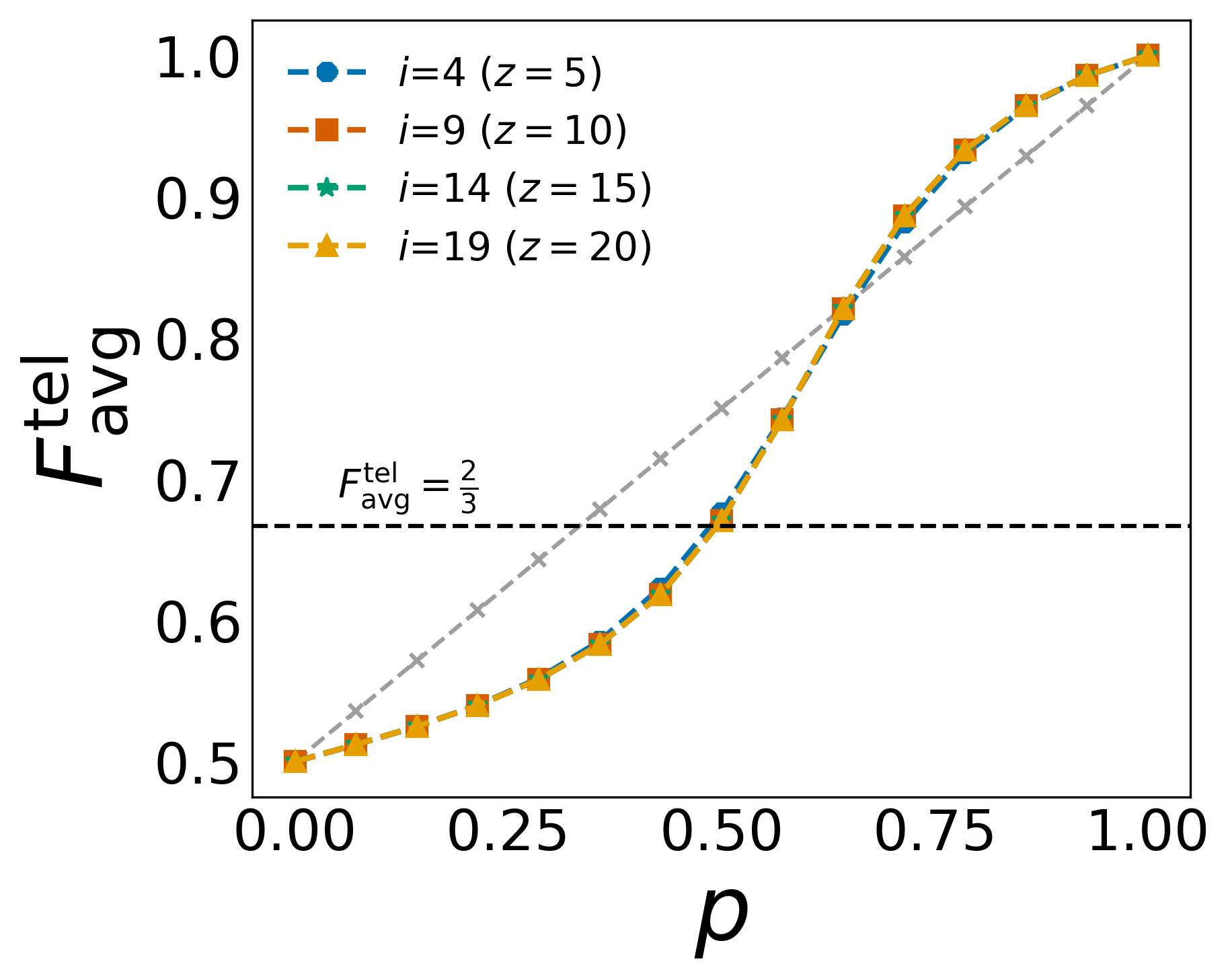}
    \\[1ex]
    \hspace*{0.75cm}(a) \hspace{3.9cm} (b)
    \caption{Plots of $F^{\mathrm{tel}}_{\mathrm{avg}}$ as a function of $p$ with varying $z$ for a complete graph network with $N=100$ where (a) $k=1$ and (b) $k=k_{\max}$.}
    \label{Fvsp_CG_saturation}
\end{figure}\\
%

\subsection{Triangular Lattice Network (TLN)}
The structure of the triangular lattice network (TLN) is based on a triangular grid in which each node has degree 6, except at the boundaries. In the present work, we consider an infinite lattice, thereby preserving full translational symmetry and neglecting boundary effects. For sufficiently large finite TLN, these boundary effects become negligible.
\begin{figure}[h]
    \centering
    \includegraphics[scale=0.15]{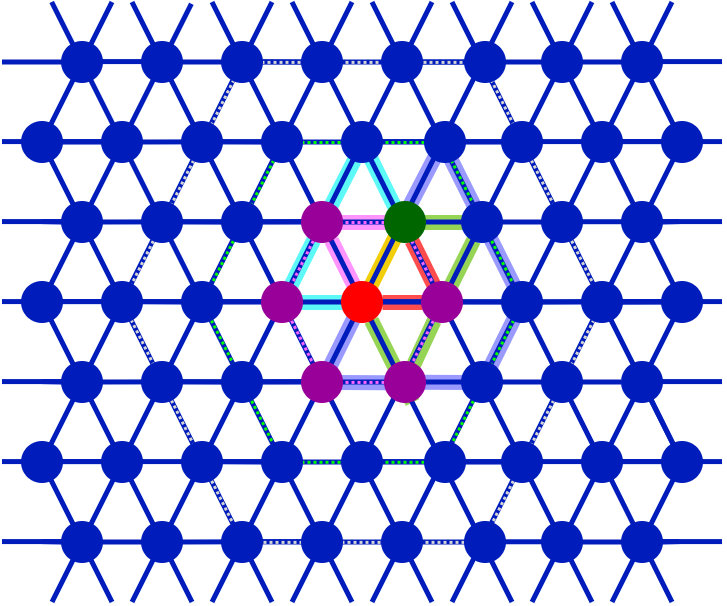}
    \hspace{0.5cm}
    \includegraphics[scale=0.15]{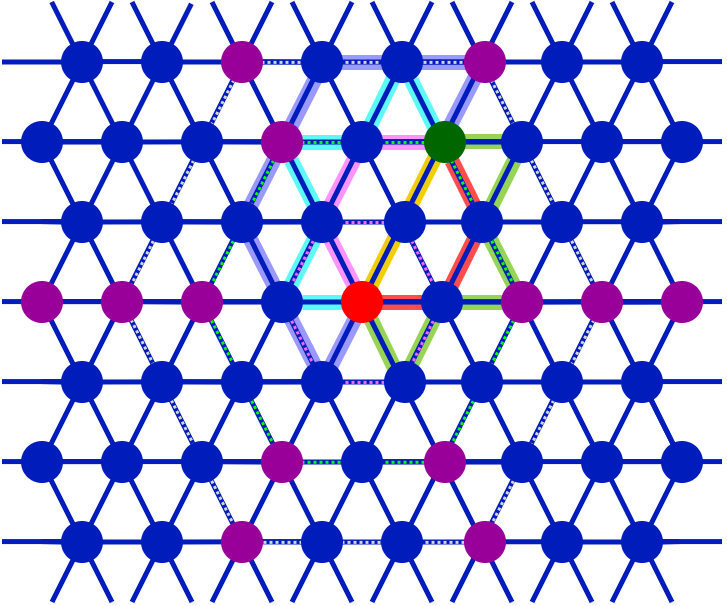}
    \\[1ex]
    (a) Nearest Neighbor \hspace{2cm} (b) Linear
    \\[1ex]
    \includegraphics[scale=0.15]{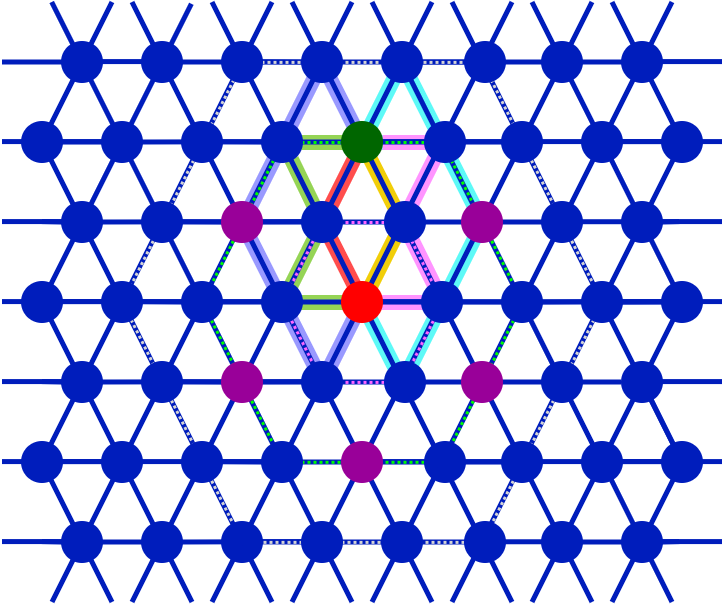}
    \hspace{0.5cm}
    \includegraphics[scale=0.15]{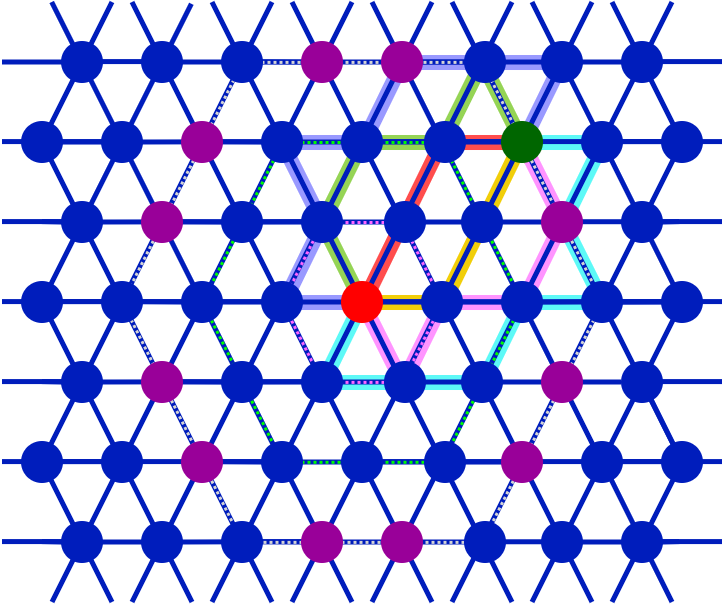}
    \\[1ex]
    (c) Diagonal \hspace{1.5cm} (d) Asymmetric Diagonal
    
    \caption{Purification procedure for an infinite triangular lattice network (TLN) for $k=1$. A fixed source node (marked in red) is considered, and the paths are determined by the position of the target node (marked in green). Four distinct classes of target positions arise, each associated with different path lengths, which are depicted in different colors for each case.}
    \label{fig:TLN}
\end{figure}
\subsubsection{One edge can be used only once, i.e. $k=1$}
As we considered the $k=1$ case, the total number of paths and the corresponding pathlengths can be determined in terms of the shortest lattice path (SLP) length $l_r$.\\
In the case of the TLN, there are four possible positions of a target $T$ with respect to a source node $S$: nearest neighbour, linear, diagonal, and asymmetric diagonal target nodes \cite{mondal2024}. Additionally, since one node has a degree of 6, the maximum number of distinct paths ($z_{\max}$) is 6. For a fixed source $S$, there are six distinct paths to any target node $T$. If $S$ and $T$ are nearest neighbours then the path lengths between them are $\mathcal{M}_{\text{NN}} = \{l_1, l_2, l_2, l_4, l_4, l_7\}$. This case arises only when $r=1$, and hence the SLP is $l_1$. \\
\indent When the target node $T$ is linear, there exist six distinct paths. However, when $r = 2$, the paths are $\mathcal{M}_{\text{L}} = \{l_2, l_3, l_3, l_6, l_6, l_{10}\}$, and for $r > 2$, the paths are $\mathcal{M}_{\text{L}} = \{l_r, l_{r+1}, l_{r+1}, l_{r+4}, l_{r+4}, l_{r+7}\}$. These paths are depicted in Fig \ref{fig:TLN}(a).  On the other hand, when $T$ is diagonal, there are also 6 possible distinct paths, i.e. $\mathcal{M}_{\text{D}} = \{l_r, l_{r}, l_{r+2}, l_{r+2}, l_{r+4}, l_{r+4}\}$ with $r$ even, as depicted in Fig. \ref{fig:TLN}(b). Finally, when $T$ is an asymmetric diagonal node, six possible paths are  $\mathcal{M}_{\text{ASD}} = \{l_r, l_{r}, l_{r+2}, l_{r+2}, l_{r+5}, l_{r+5}\}$ with $r>2$, and it is shown in Fig. \ref{fig:TLN}(c). Here, the diagonal and asymmetric diagonal case only arises when $r>1$.\\
\indent Finally, the total number of nodes is $6r$ for a given SLP $l_r$. For $r=1$, only six nearest neighbour nodes are present, and hence the fraction is $\frac{6r}{6r}=1$. In the case of $r=2$, there are six linear as well as diagonal nodes. Hence, the contributions are $\frac{6}{6r} = \frac{1}{r} = \frac{1}{2}$ for both the cases. Finally, two cases arise for $r>2$. When $r$ is even, there are six linear nodes, six diagonal nodes, and hence the total number of asymmetric diagonal nodes is $6r-12$. The contributions are $\frac{1}{r}, \frac{1}{r}$, and $\frac{r-2}{r}$ respectively. However, for odd $r$, six linear nodes are contributing $\frac{1}{r}$ as the fraction and rest $6r-6$ nodes are asymmetric diagonal providing $\frac{r-1}{r}$ as the fraction. Incorporating all the contributions from $T$ nodes and fixed $S$ node (denoted as $(S,*)$),
\begin{eqnarray*}
    &&F^{\mathrm{tel}}_{r}|_{(S,*)} \nonumber \\
    &&= 
    \begin{cases}
    F^{\mathrm{tel}}(\mathcal{M}_{\text{NN}}|^{(z)}_{\text{SPL}}), \qquad \qquad \qquad \qquad \qquad \text{for $r=1$,}\\
    \frac{1}{2}F^{\mathrm{tel}}(\mathcal{M}_{\text{L}}|^{(z)}_{\text{SPL}}) +\frac{1}{2}F^{\mathrm{tel}}(\mathcal{M}_{\text{D}}|^{(z)}_{\text{SPL}}), \qquad \text{for $r=2$,}\\
    \frac{1}{r} F^{\mathrm{tel}}(\mathcal{M}_{\text{L}}|^{(z)}_{\text{SPL}}) +\frac{r-1}{r}F^{\mathrm{tel}}(\mathcal{M}_{\text{ASD}}|^{(z)}_{\text{SPL}}), \text{for $r$ = odd,}\\
    \frac{1}{r} F^{\mathrm{tel}}(\mathcal{M}_{\text{L}}|^{(z)}_{\text{SPL}})
    +\frac{1}{r} F^{\mathrm{tel}}(\mathcal{M}_{\text{D}}|^{(z)}_{\text{SPL}}) \\    
    +\frac{r-2}{r}F^{\mathrm{tel}}(\mathcal{M}_{\text{ASD}}|^{(z)}_{\text{SPL}}), \qquad \qquad \qquad \quad \text{for $r$ = even.}
    \end{cases}
\end{eqnarray*}
%
To compute the average, we have to consider all $(S,T)$ pairs. Hence,
\begin{eqnarray}
    \Ftel = \sum_{r=1}^{r_{\max}} p_r F_r^{\mathrm{tel}}|_{(S,*)},
\end{eqnarray}
where $p_r$ denotes the fraction of node pairs separated by distance $r$ in the TLN, and the maximum value of $r$ is $r_{\max} = 2(n-1)$ if we consider the $N = n \times n$ triangular grid. The exact value of $p_r$ is difficult to calculate analytically.

\subsubsection{One edge can be used maximum number of times i.e. $k=k_{\max}$}
In this scenario, the longest path length for purification for a fixed source can extend up to $2(n-1)$, i.e. $\mathcal{O}(n)$, making it extremely difficult to determine the number of paths analytically. Because many alternative intermediary node configurations contribute to paths of different lengths, the combinatorial complexity rises quickly with the network size. As a result, it is difficult to get a closed-form equation for the average of teleportation fidelity. In order to assess the relevant quantities and derive the related outcomes, we turn to numerical results.

\subsubsection{Results}
We plot $\Ftel$ as a function of visibility $p$ for entanglement distribution using both MPEP and without MPEP scenarios on finite triangular lattices with equal numbers of rows and columns, i.e., $m=n$. We first focus on the case $k=1$. For $(m,n)=(10,10)$ and a maximum number of purifications $i=5$, the MPEP strategy outperforms the scenario without MPEP throughout the entire range of $p$. This result demonstrates that the MPEP strategy provides a significant improvement over without MPEP scenario, as shown in Fig. \ref{Fvsp_TLN_k1}(a). 
The corresponding relative gains are shown in Fig. \ref{Fvsp_TLN_k1}(b). In addition to that, the proposed algorithm can attain the quantum advantage at $p_c \approx 0.8793, 0.8589, 0.8430, 0.8299,$ and $0.8186$ for $i=1, 2, 3, 4$ and $5$ respectively. Whereas without MPEP scenario suggests that the quantum advantage can be achieved at $p_c \approx 0.9231$. Therefore, the MPEP protocol needs to be implemented to achieve the quantum advantage at a relatively low value of $p_c$ in the case of a SLN.
\begin{figure}[h]
    \centering
    \includegraphics[scale=0.27]
        {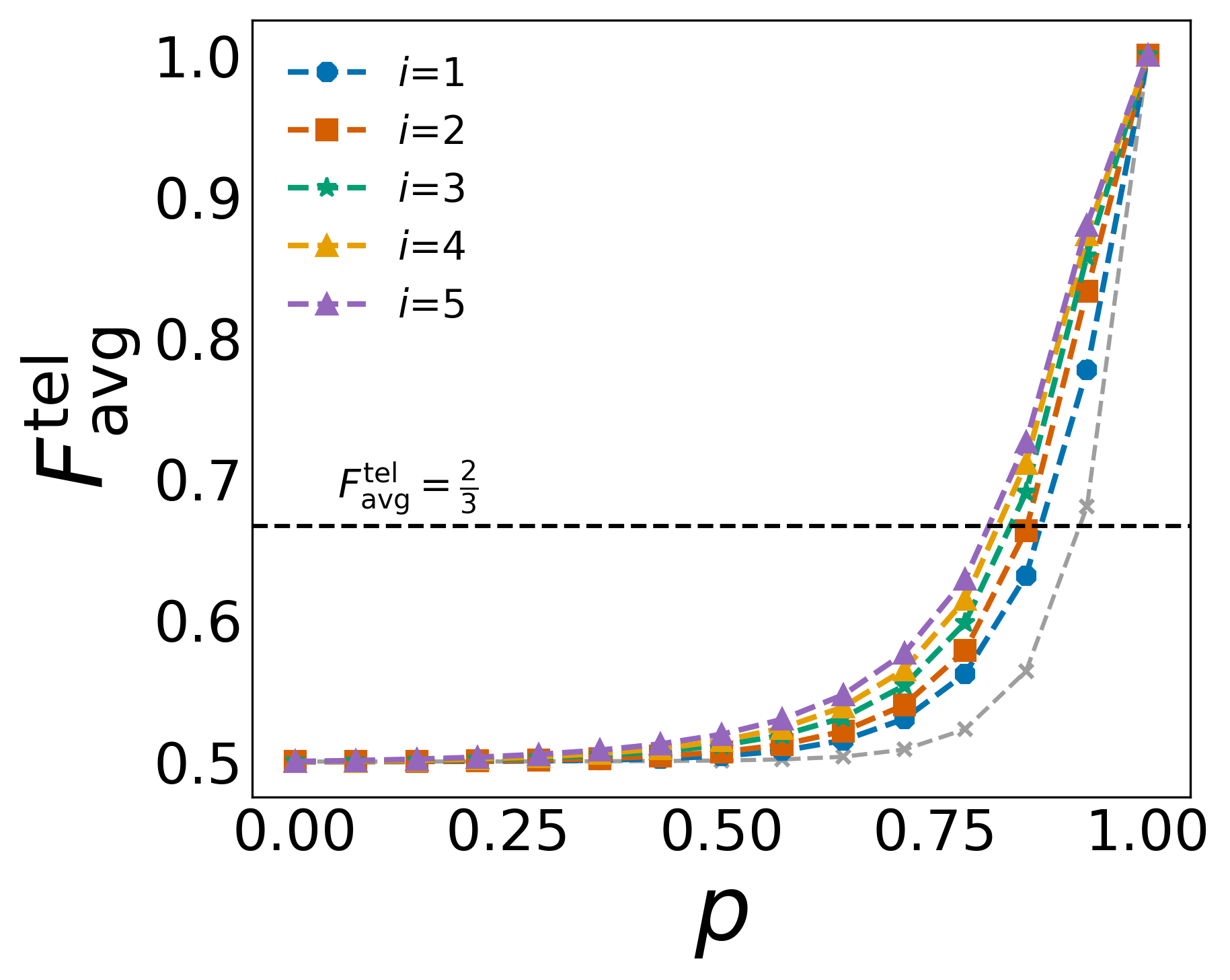}
    \includegraphics[scale=0.27]
        {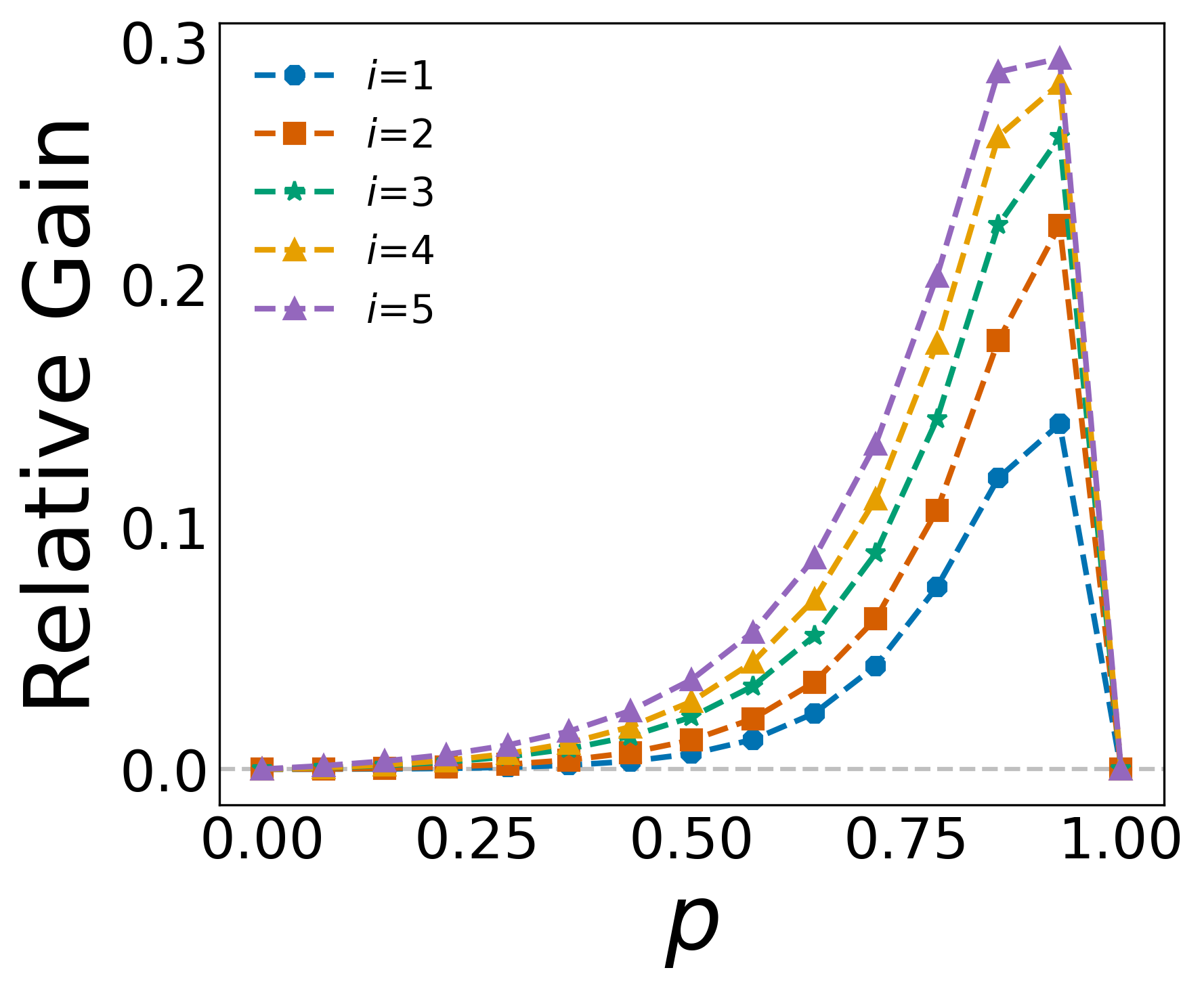}
    \\[1ex]
    \hspace*{0.75cm}(a) \hspace{3.9cm} (b)

    \caption{Plots of (a) $F^{\mathrm{tel}}_{\mathrm{avg}}$ as a function of $p$ and (b) relative gain as a function of $p$ for a triangular lattice network with $(m,n)=(10,10)$ and $k=1$. The total number of paths is taken to be $z=z_{\max}=6$. The horizontal black line suggests $\Ftel=\frac{2}{3}$. The grey line denotes the variation of $\Ftel$ with respect to $p$ for the without MPEP scenario.}
    \label{Fvsp_TLN_k1}
\end{figure}\\
\indent For $k=k_{\max}$ scenario with the same configuration, the enhancement can be achieved at $p \approx 0.1808, 0.2019, 0.1866$ for $i=1, 5$, and 9, respectively. For $i=19$, the proposed strategy shows the improvement throughout the entire range of $p$. These results are depicted in Fig. \ref{Fvsp_TLN_kmax}(a) and the corresponding relative gains are plotted in Fig. \ref{Fvsp_TLN_kmax}(b). Furthermore, the suggested approach may achieve the quantum advantage at $p_c \approx 0.8409, 0.8351, 0.8307,$ and $0.7997$ for $i=1, 5, 10,$ and $19$, respectively. These results indicate that the MPEP protocol provides a significant advantage.
\begin{figure}[h]
    \centering
    \includegraphics[scale=0.27]
        {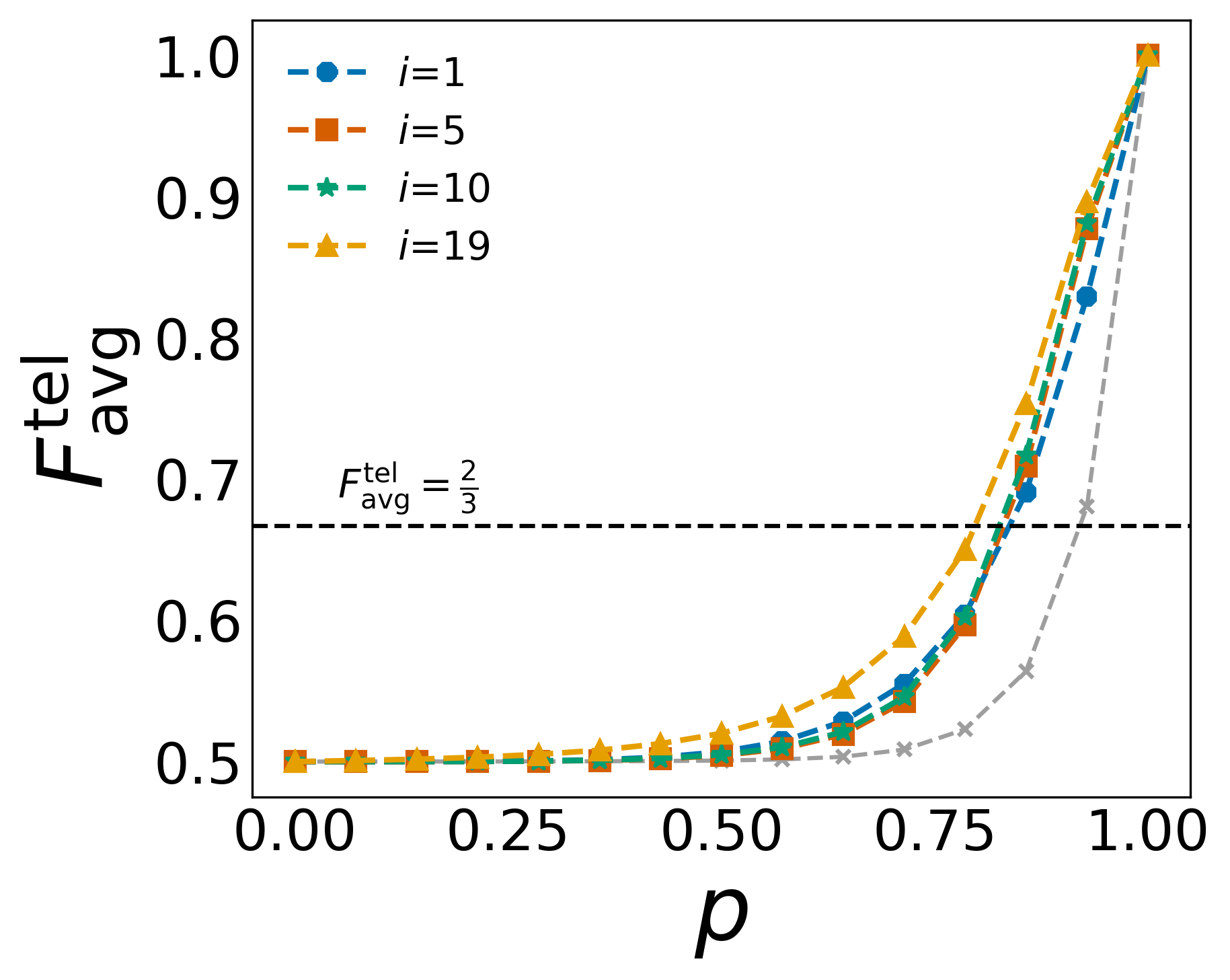}
    \includegraphics[scale=0.27]
        {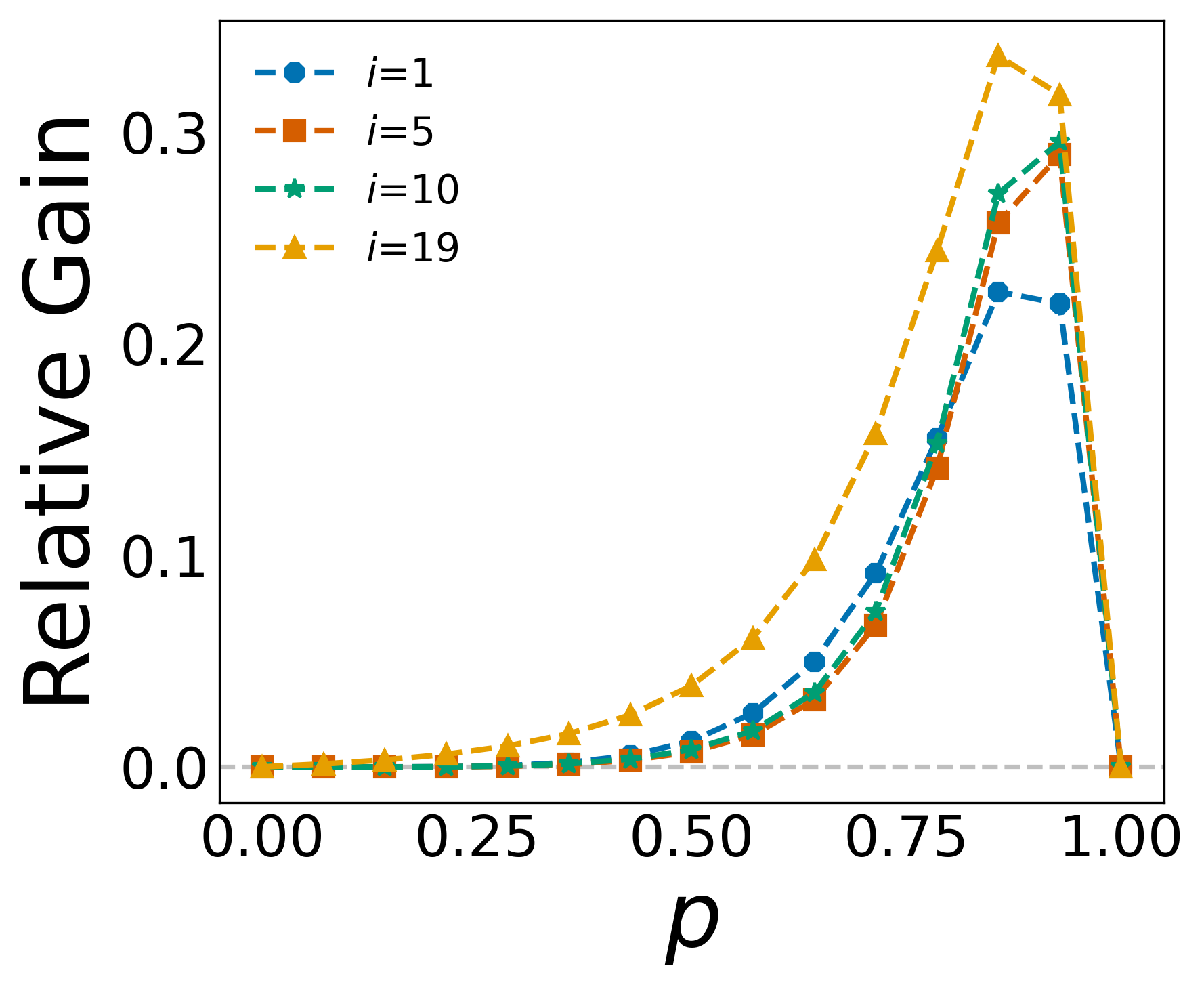}
    \\[1ex]
    \hspace*{0.75cm}(a) \hspace{3.9cm} (b)

    \caption{Plots of (a) $F^{\mathrm{tel}}_{\mathrm{avg}}$ as a function of $p$ and (b) relative gain as a function of $p$ for a triangular lattice network with $(m,n)=(10,10)$ and $k=k_{\mathrm{max}}$. The total number of paths is taken to be $z=20$. The horizontal black line suggests $\Ftel=\frac{2}{3}$. The grey line denotes the variation of $\Ftel$ with respect to $p$ for the without MPEP scenario.}
    \label{Fvsp_TLN_kmax}
\end{figure}

As $z$ increases, the $\Ftel$ gradually reaches its saturation value, and after a certain number of $z$, where $i_{\max}=z-1$. For $(m,n)=(10,10)$ and $k=1$ case, the maximum value of $z$ is 6. Hence, $\Ftel$ already reaches its maximum value in this case. However, we observed that the change in $\Ftel$ is around $10^{-4}$ after $z=6$ and $i=5=i_{\max}$ for the same TLN but for the case $k=k_{\max}$. Fig. \ref{Fvsp_TLN_saturation} shows the saturation of $\Ftel$ as $z$ increases. As a result, we can conclude that using one edge multiple times has no noticeable advantage over using each edge only once.
\begin{figure}[h]
    \centering
    \includegraphics[scale=0.27]{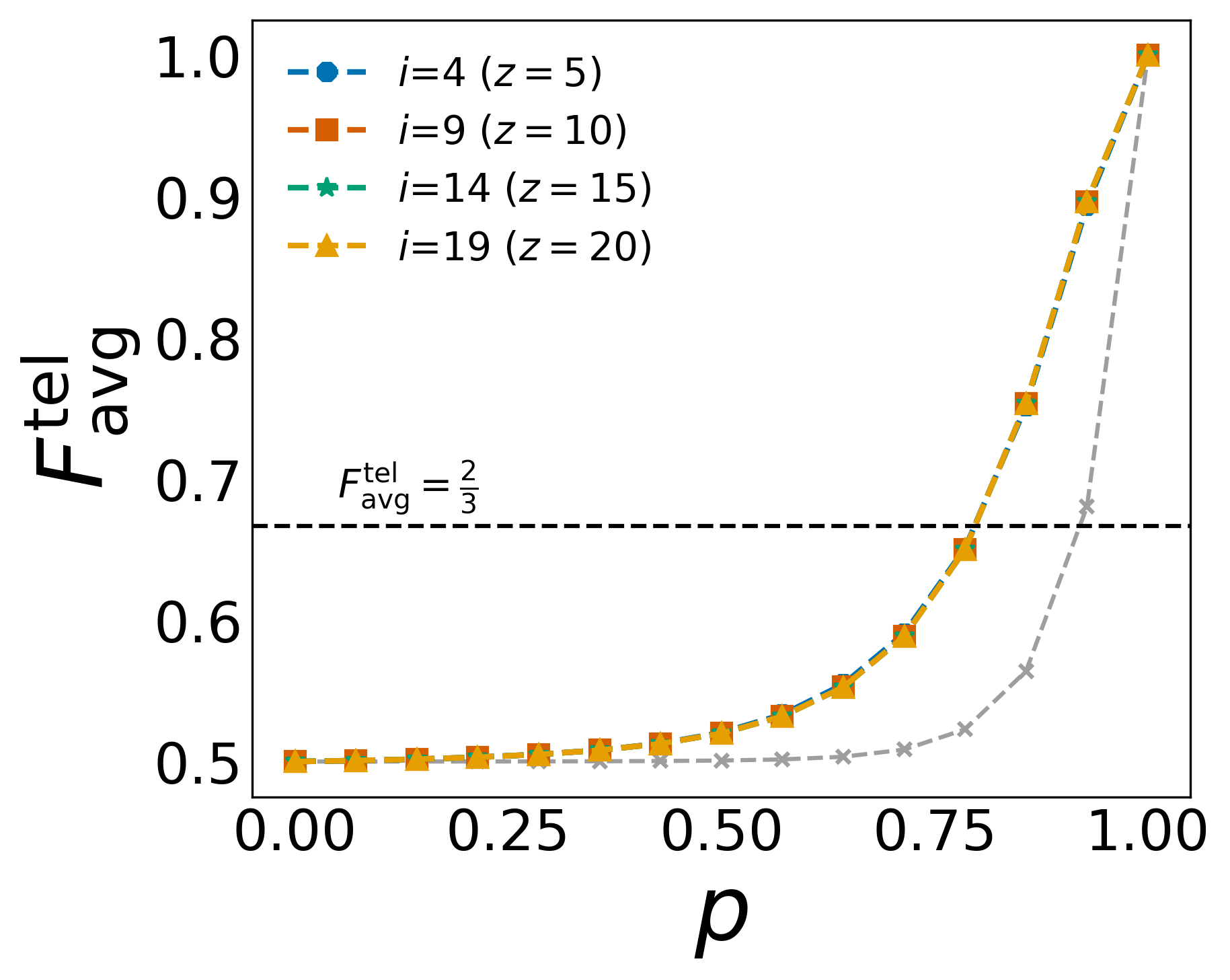}  
    \caption{Plots of $F^{\mathrm{tel}}_{\mathrm{avg}}$ as a function of $p$ with varying $z$ for a triangular lattice network network with $(m,n)=(10,10)$ and $k=k_{\max}$.}
    \label{Fvsp_TLN_saturation}
\end{figure}
%

\subsection{Square Lattice Network (SLN)}
Every node in the square grid that forms the basis of the square lattice network (SLN) has degree 4, except for the boundary. Just like the previous case, we omit border effects and retain the full symmetry by considering an infinite lattice. The boundary effects become insignificant for sufficiently large finite SLN. 

\begin{figure}[h]
    \centering
    \includegraphics[scale=0.15]{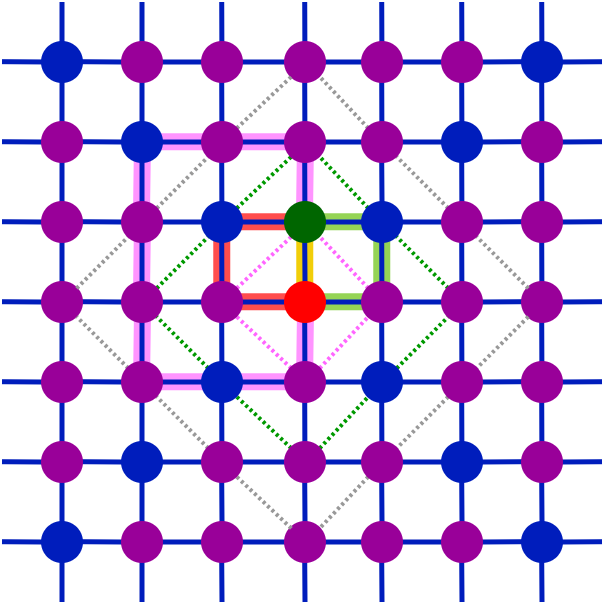}
    \hspace{1cm}
    \includegraphics[scale=0.15]{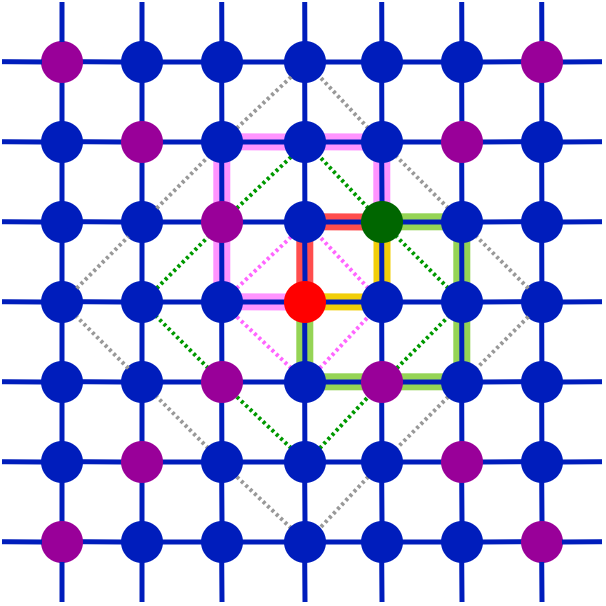}
    \\[1ex]
    (a) Non-diagonal \hspace{2cm} (b) Diagonal
    
    \caption{Purification procedure for an infinite square lattice network (SLN) for $k=1$. A fixed source node (marked in red) is considered, and the paths are determined by the position of the target (marked in green) node. Two distinct classes of target positions arise, each associated with different path lengths, which are depicted in different colours for each case.}
    \label{fig:SLN}
\end{figure}

\subsubsection{One edge can be used only once, i.e. $k=1$}
As we considered the $k=1$ case, the number of distinct paths and the corresponding path lengths can be determined in terms of the shortest lattice path (SLP) length $l_r$. Additionally, the maximum number of paths ($z_{\max}$) is 4 here because one node has a degree of 4.\\
%
\indent For the SLN, two distinct cases arise, corresponding to diagonal and non-diagonal target nodes \cite{mondal2024}. For a fixed source node $S$, when the target node $T$ is diagonal, there exist four distinct paths, $\mathcal{M}_{\mathrm{D}} = \{l_r, l_{r}, l_{r+4}, l_{r+4}\}$ as shown in Fig \ref{fig:SLN}(a).  On the other hand, when $T$ is non-diagonal, there are also 4 possible paths, i.e. $\mathcal{M}_{\mathrm{AD}} = \{l_r, l_{r+2}, l_{r+2}, l_{r+8}\}$ as depicted in Fig. \ref{fig:SLN}(b).\\
\indent Moreover, for a given SLP $l_r$, the total number of nodes is $4r$. The SLN exhibits a peripheral symmetry, i.e. for odd $r$, only non-diagonal nodes contribute, and hence the number of non-diagonal nodes is $4r$. The fraction of non-diagonal nodes is $\frac{4r}{4r}=1$. Whereas for even $r$, both the diagonal and non-diagonal nodes are present. The total number of diagonal nodes in this case is $4$ (which is always fixed for a given $l_r$), and therefore the total number of non-diagonal nodes is $4r-4$. Hence the fraction of diagonal nodes is $\frac{4}{4r}=\frac{1}{r}$, and fraction of non-diagonal nodes is $\frac{4r-4}{4r}=\frac{r-1}{r}$. If $z$ represents the maximum number of paths utilised in the purifying process, and if the paths are ordered in descending order in accordance with the SPL strategy, we have
%
\begin{eqnarray*}
    F^{\mathrm{tel}}_{r}|_{(S,*)} = 
    \begin{cases}
       \frac{1}{r} F^{\mathrm{tel}}(\mathcal{M}^{(z)}_{\mathrm{D}}|_{\mathrm{SPL}}) \\ +\frac{r-1}{r}F^{\mathrm{tel}}(\mathcal{M}^{(z)}_{\mathrm{AD}}|_{\mathrm{SPL}}), \quad \text{for $r$ = even,}\\
       F^{\mathrm{tel}}(\mathcal{M}^{(z)}_{\mathrm{AD}}|_{\mathrm{SPL}}), \qquad \quad \text{ for $r$ = odd}.
    \end{cases}
\end{eqnarray*}
Here $(S,*)$ suggests that this is valid for a fixed source node. To compute the average, we must include all $(S,T)$ pairs. Hence,
\begin{eqnarray}
    \Ftel = \sum_{r=1}^{r_{\max}} p_r F_r^{\mathrm{tel}}|_{(S,*)},
\end{eqnarray}
where $p_r$ denotes the fraction of node pairs separated by distance $r$ in the SLN, and the maximum value of $r$ is $r_{\max} = 2(n-1)$ if we consider the $N = n \times n$ square grid. The exact value of $p_r$ is difficult to calculate analytically. \\ \\


\subsubsection{One edge can be used maximum number of times i.e. $k=k_{\max}$}
It is quite challenging to calculate the number of unique paths analytically similar to the TLN, since the longest path length for purification for a fixed source can reach up to $2(n-1)$, i.e. $\mathcal{O}(n)$. Consequently, obtaining a closed-form formula for the average teleportation fidelity is challenging. We use numerical results to evaluate the pertinent values and obtain the associated findings.
\subsubsection{Results}
On finite square lattices with equal numbers of rows and columns, i.e., $m=n$, we plot $\Ftel$ as a function of visibility $p$ for both MPEP and without MPEP scenarios. First, we concentrate on the $k=1$ situation. For $(m,n)=(10,10)$ with the maximum number of purification $z=4$, we obtain an advantage using the MPEP protocol over the case where MPEP is not used. It is illustrated in Fig. \ref{Fvsp_SLN_k1}(a). Fig. \ref{Fvsp_SLN_k1}(b) which display the associated relative gains. Furthermore, we found that the MPEP protocol can achieve the quantum advantage at relatively lower values of $p$. The threshold for achieving the quantum advantage is at $p_c \approx 0.8849, 0.8610$, and $0.8494$ for $i=1, 2,$ and $3$, respectively. In contrast, the without MPEP scenario exhibits quantum advantage at $p_c \approx 0.9249$ which much larges than that when MPEP is employed.
\begin{figure}[H]
    \centering
    \includegraphics[scale=0.27]
        {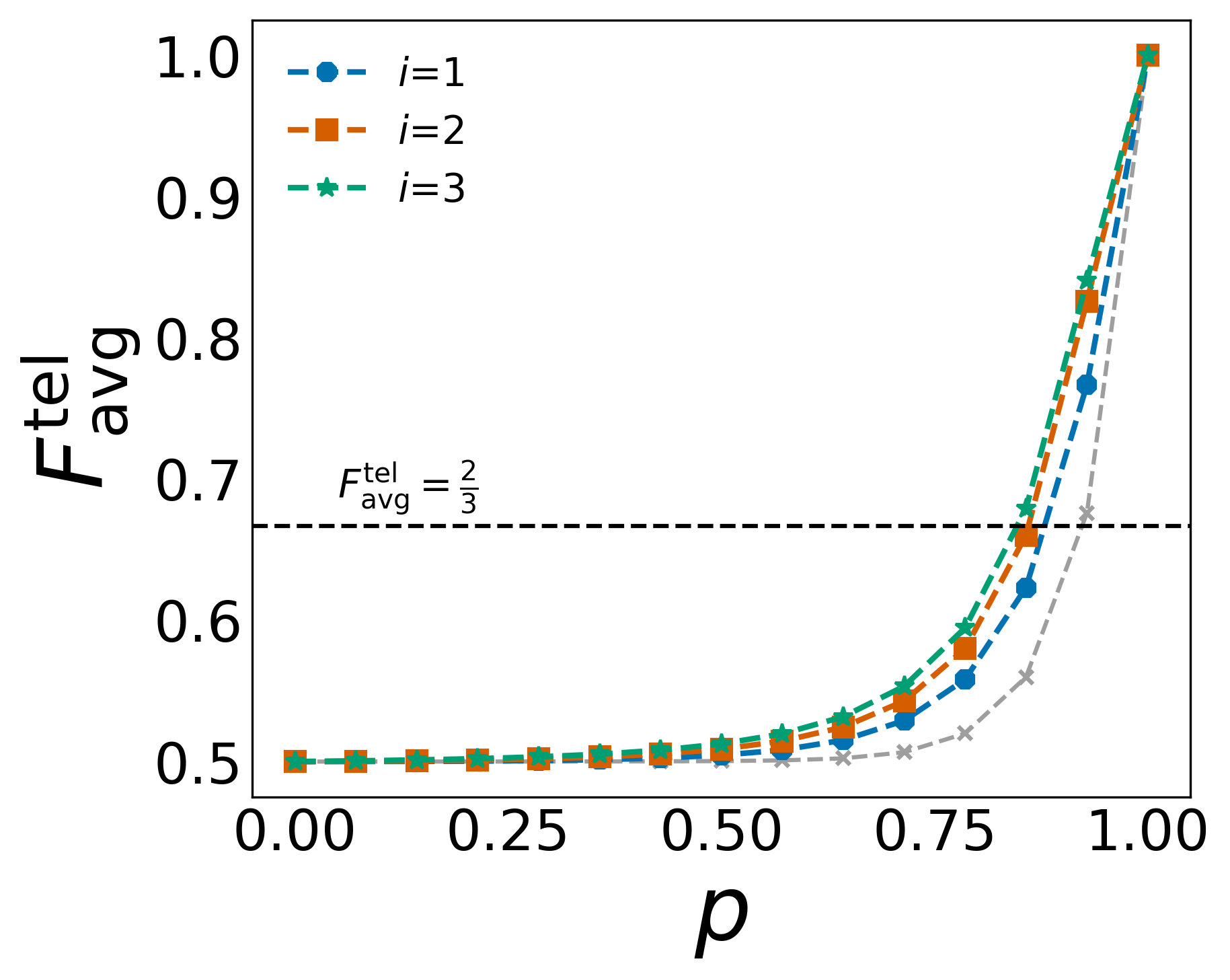}
    \includegraphics[scale=0.27]
        {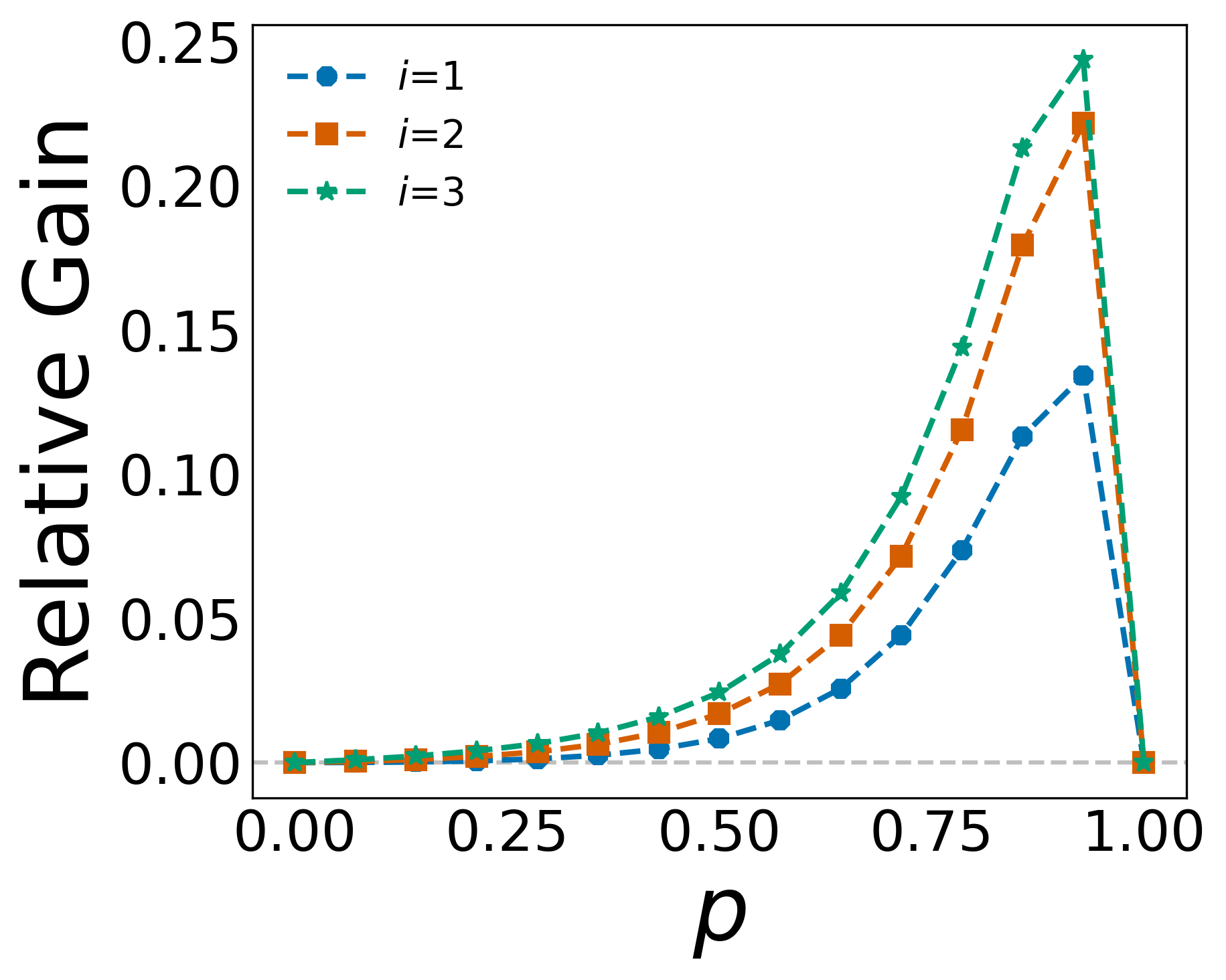}
    \\[1ex]
    \hspace*{0.75cm}(a) \hspace{3.9cm} (b)

    \caption{Plots of (a) $F^{\mathrm{tel}}_{\mathrm{avg}}$ as a function of $p$ and (b) relative gain as a function of $p$ for a square lattice network with $(m,n)=(10,10)$ and $k=1$. The total number of paths is taken to be $z=z_{\max}=4$. The horizontal black line suggests $\Ftel=\frac{2}{3}$. The grey line denotes the variation of $\Ftel$ with respect to $p$ for the without MPEP scenario.}
    \label{Fvsp_SLN_k1}
\end{figure}
\indent For the $k=k_{\max}$ situation, we can also achieve the enhancement for the same SLN at $p \approx 0.2895, 0.3332$, and $0.8118$ for $i=1, 5$, and $10$, respectively. However, for $i=19$, it immediately shows the improvement over the case when MPEP is not used. Therefore, the transition value will be lower with higher number of paths used in the MPEP protocol. The results along with the relative gains are shown in Fig. \ref{Fvsp_SLN_kmax}(a) and Fig. \ref{Fvsp_SLN_kmax}(b), respectively. These findings reveal that the MPEP protocol offers a notable improvement. Moreover, we discover that the proposed algorithm gains a quantum advantage at relatively lower values of $p$. The threshold for achieving the quantum advantage is at $p_c \approx 0.8727, 0.8705, 0.8646,$ and $0.8319$ for $i=1,5,10,$ and $19$, respectively.
\begin{figure}[h]
    \centering
    \includegraphics[scale=0.27]
        {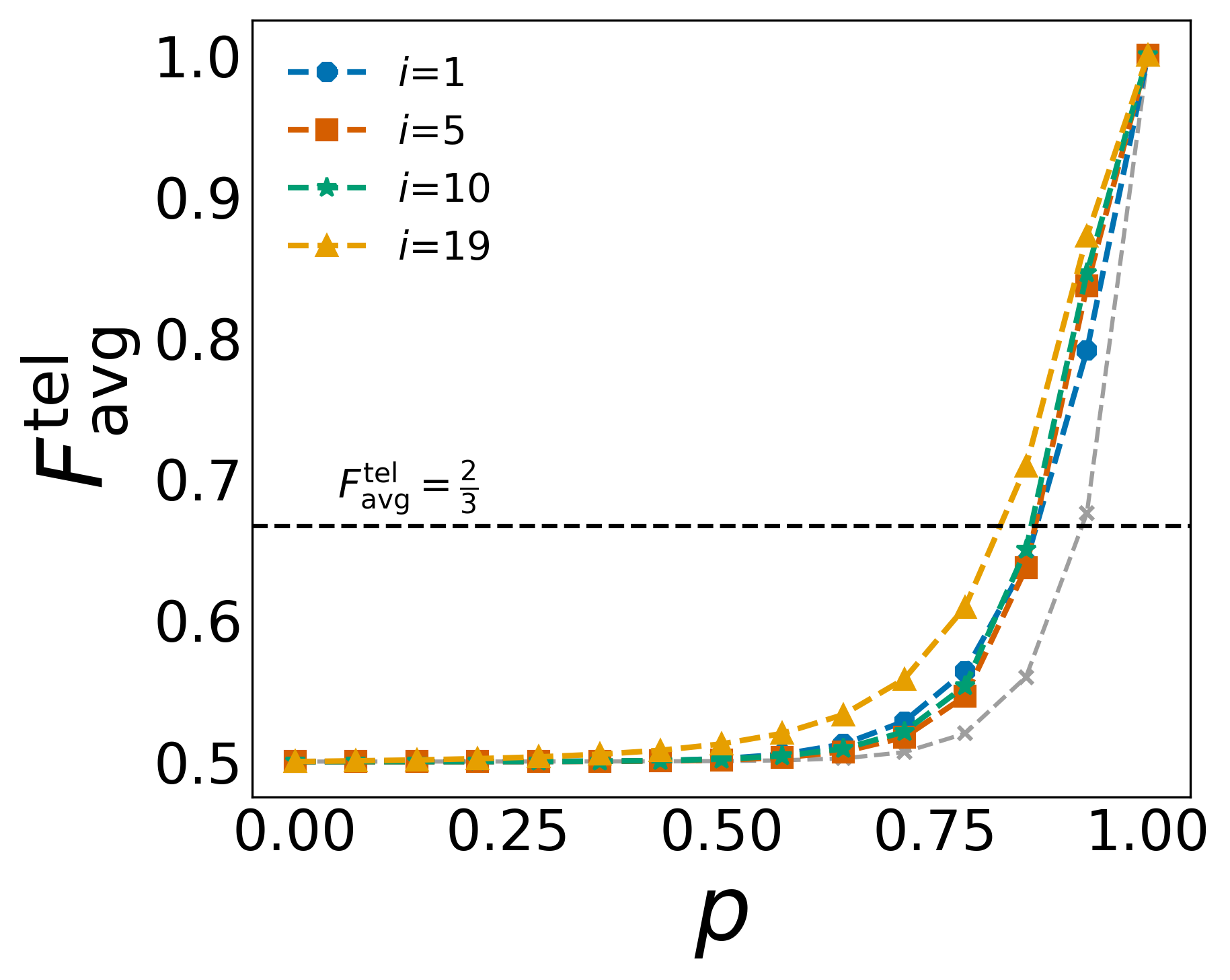}
    \includegraphics[scale=0.27]
        {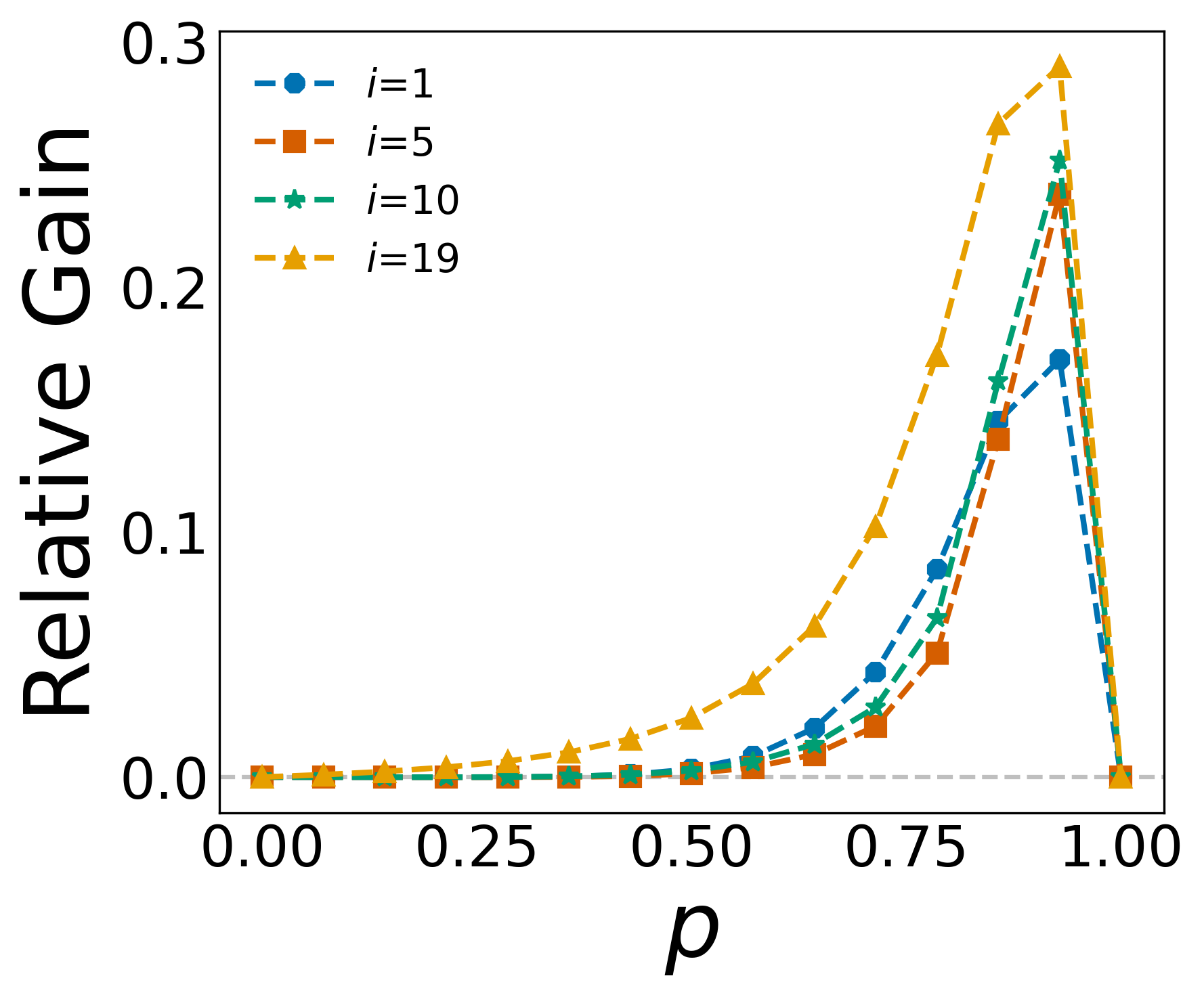}
    \\[1ex]
    \hspace*{0.75cm}(a) \hspace{3.9cm} (b)

    \caption{Plots of (a) $F^{\mathrm{tel}}_{\mathrm{avg}}$ as a function of $p$ and (b) relative gain as a function of $p$ for a square lattice network with $(m,n)=(10,10)$ and $k=k_{\mathrm{max}}$. The total number of paths is taken to be $z=20$. The horizontal black line suggests $\Ftel=\frac{2}{3}$. The grey line denotes the variation of $\Ftel$ with respect to $p$ for the without MPEP scenario.}
    \label{Fvsp_SLN_kmax}
\end{figure}\\
\indent Similar to previous topologies, when $z$ increases, the $\Ftel$ gradually approaches its saturation value, and after a given number of $z$, where $i_{\max}=z-1$. For a SLN with $(m,n)=(10,10)$ and $k=1$, the maximum value of $z$ is 4. Therefore, $\Ftel$ has already reached its maximum value in this scenario. Moreover, we found that the change in $\Ftel$ is roughly $10^{-4}$ after $z=7$ and $i=6=i_{\max}$ for the same SLN, but for the case $k=k_{\max}$. Fig. \ref{Fvsp_SLN_saturation} illustrates how $\Ftel$ saturates as $z$ rises. Consequently, we can conclude that increasing the number of $z$ and $i$ does not significantly affect $\Ftel$.
\begin{figure}[h]
    \centering
    \includegraphics[scale=0.27]{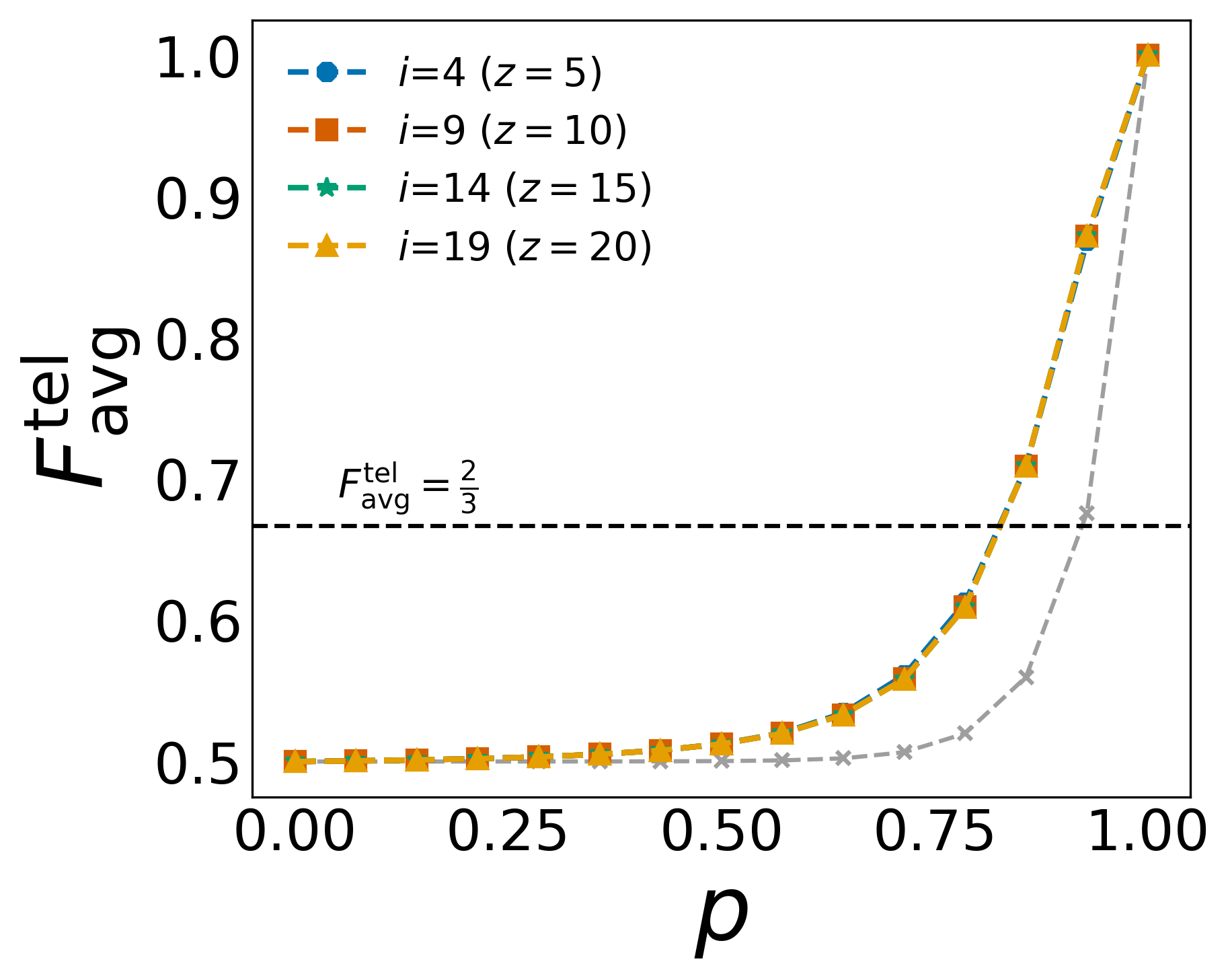}  
    \caption{Plots of $F^{\mathrm{tel}}_{\mathrm{avg}}$ as a function of $p$ with varying $z$ for a square lattice network network with $(m,n)=(10,10)$ and $k=k_{\max}$.}
    \label{Fvsp_SLN_saturation}
\end{figure}
%
\section{Purification with Irregular topologies}
\label{sec5}
In this section, we investigate the feasibility of MPEP protocol for enhancing the average teleportation fidelity of irregular network topologies.  We present the scaling of the average teleportation fidelity and relative gain with respect to visibility $p$ of the entangled states distributed along the network edges, and provide a detailed comparison of the utilization of MPEP protocol over single path entangelement distribution protocol without MPEP.\\
\indent An example of irregular topology is a random network. A random network is a network in which the edges between nodes are generated at random according to a specified probability distribution. One of the related models of generating a random network is the Erd\H{o}s-R\'enyi network (ERN) \cite{erdos1959}. An ERN can be constructed by connecting labelled nodes randomly, where an edge is included in the graph with a probability $q \in [0,1]$, independently of every other edge. However, this type of graph disagrees with the real-world networks, which often tend to follow the scale-free distribution.
\begin{figure}[h]
    \centering
    \includegraphics[scale=0.27]{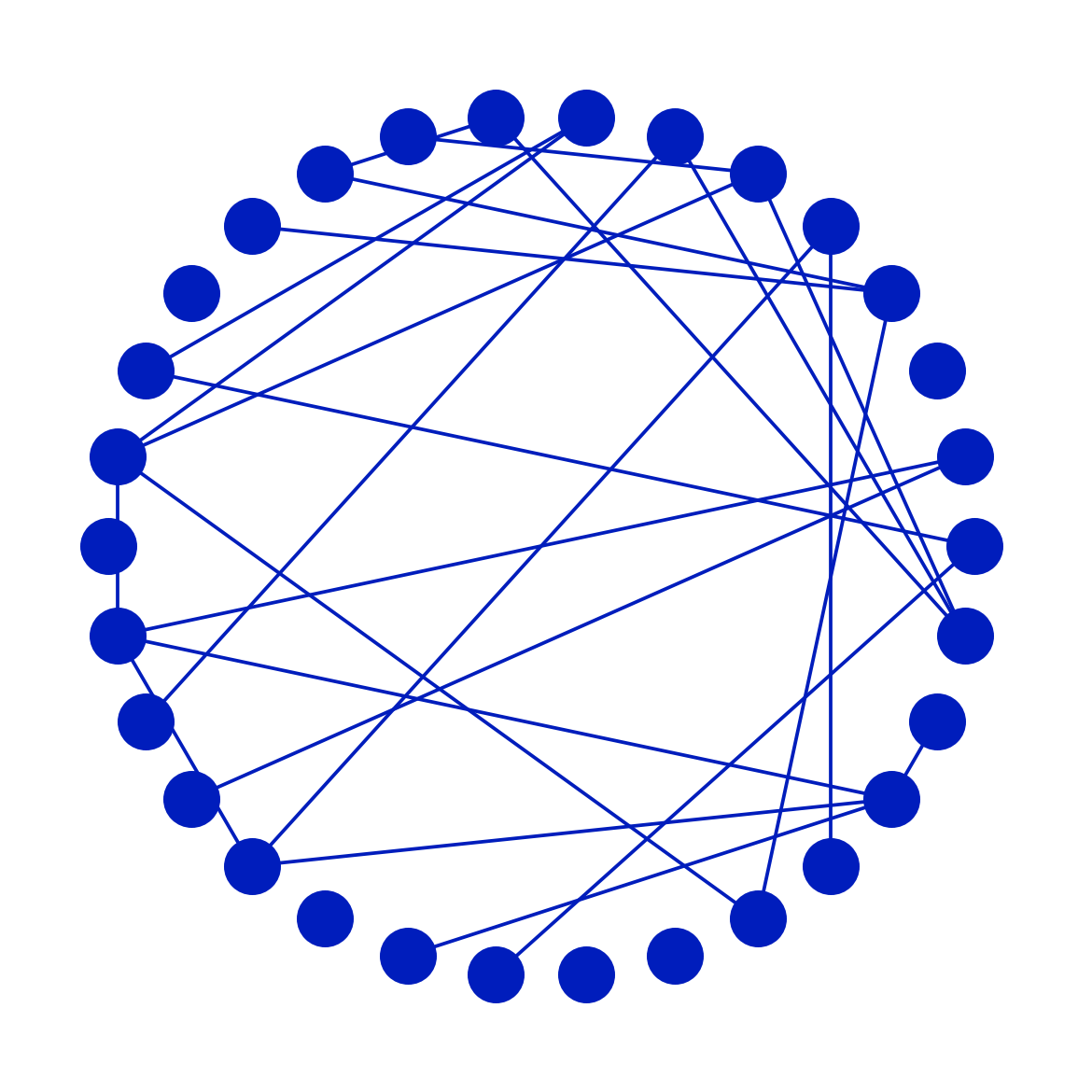}
    \includegraphics[scale=0.27]{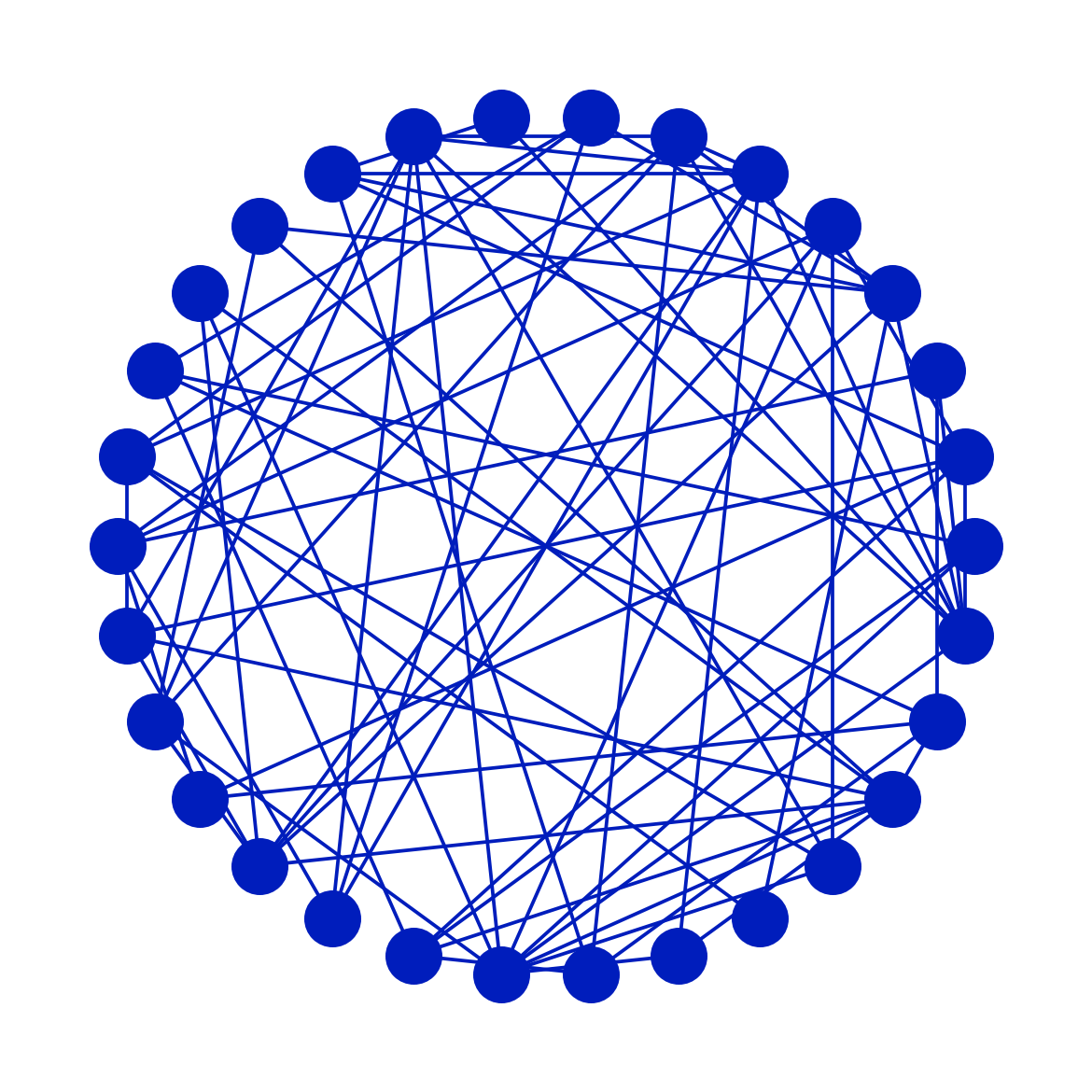}
    \includegraphics[scale=0.27]{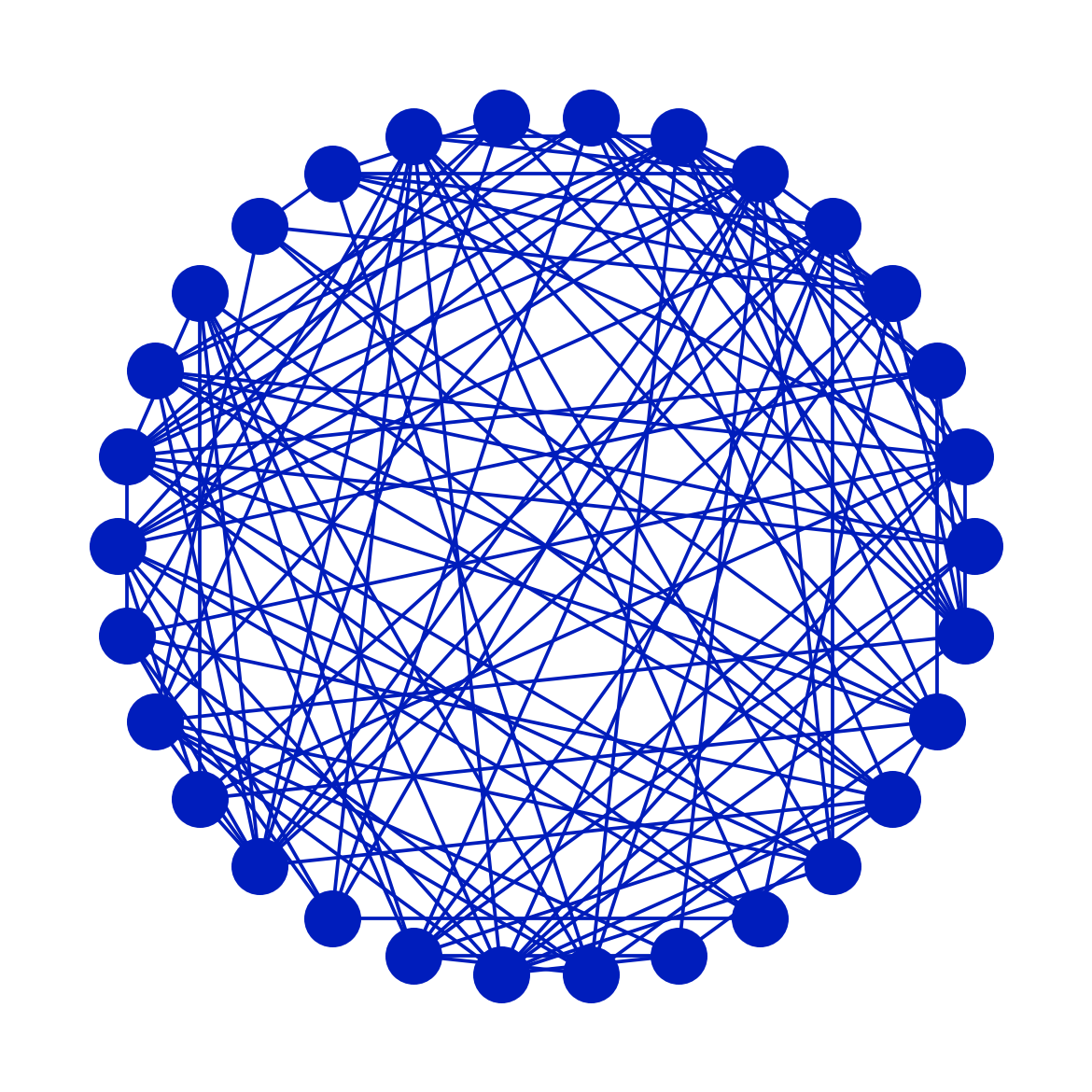}
    \\[1ex]
    \hspace*{0.08cm}(a) $E=25$ \hspace{1.1cm} (b) $E=75$ \hspace{1.0cm} (c) $E=125$
    
    \caption{Erd\H{o}s--R\'enyi network (ERN) with $N=30$. The graphs correspond to a single realization with varying numbers of edges.}
    \label{fig:ER}
\end{figure}\\
\indent There are two standard representations of ERN: $G(N,q)$, where $N$ denotes the number of nodes and $q$ is the edge probability, and $G(N,E)$, where $E$ is the total number of edges. In the $G(N,q)$ model, the number of edges varies across different realization of the graph, whereas in the case of $G(N,E)$ model, the total number of edges is fixed. In this article, we adopt the $G(N,E)$ model because it is more suitable for implementing our proposed algorithms. One example is shown in Fig. \ref{fig:ER}, where the edge probability $q$ is varied for a fixed number of nodes.\\

%
\subsubsection*{Results}
Due to the randomness of the graph, it is difficult to achieve a closed-form expression for $k=1$ as well as $k=k_{\max}$ case. Therefore, we turned to the numerical results for this topology. To include this statistical robustness, we generated the ERN with a fixed number of nodes and edges, 25 times $(T_1,T_2,\hdots,T_{25})$ for each fixed $q$, and for each time the corresponding $\Ftel$ is calculated. We then evaluated the average of all these trials to obtain $\langle\Ftel\rangle_{T_1,\hdots,T_{25}}$.\\
\indent Let us concentrate $k=1$ case at first. Here, we vary the total number of paths $z$ for each instance, and we consider the number of purifications $i = i_{\max} = z-1$. We find that the average teleportation fidelity $\Ftel$ with MPEP approach surpasses that when MPEP is not employed at the values $p \approx 0.7191, 0.7232, 0.7284, 0.73,$ and $ 0.7305$ for $z = 2, 3, 4, 5$, and $6$ respectively. The value of $\Ftel$ achieves the saturation after $z=5$. These findings, along with the variation of relative gains, as defined on Eq. \eqref{rg}, are shown in Fig. \ref{Fvsp_ERN_k1}(a) and \ref{Fvsp_ERN_k1}(b). These findings show that the proposed algorithms improve $\Ftel$.   
\begin{figure}[h]
    \centering
    \includegraphics[scale=0.27]
        {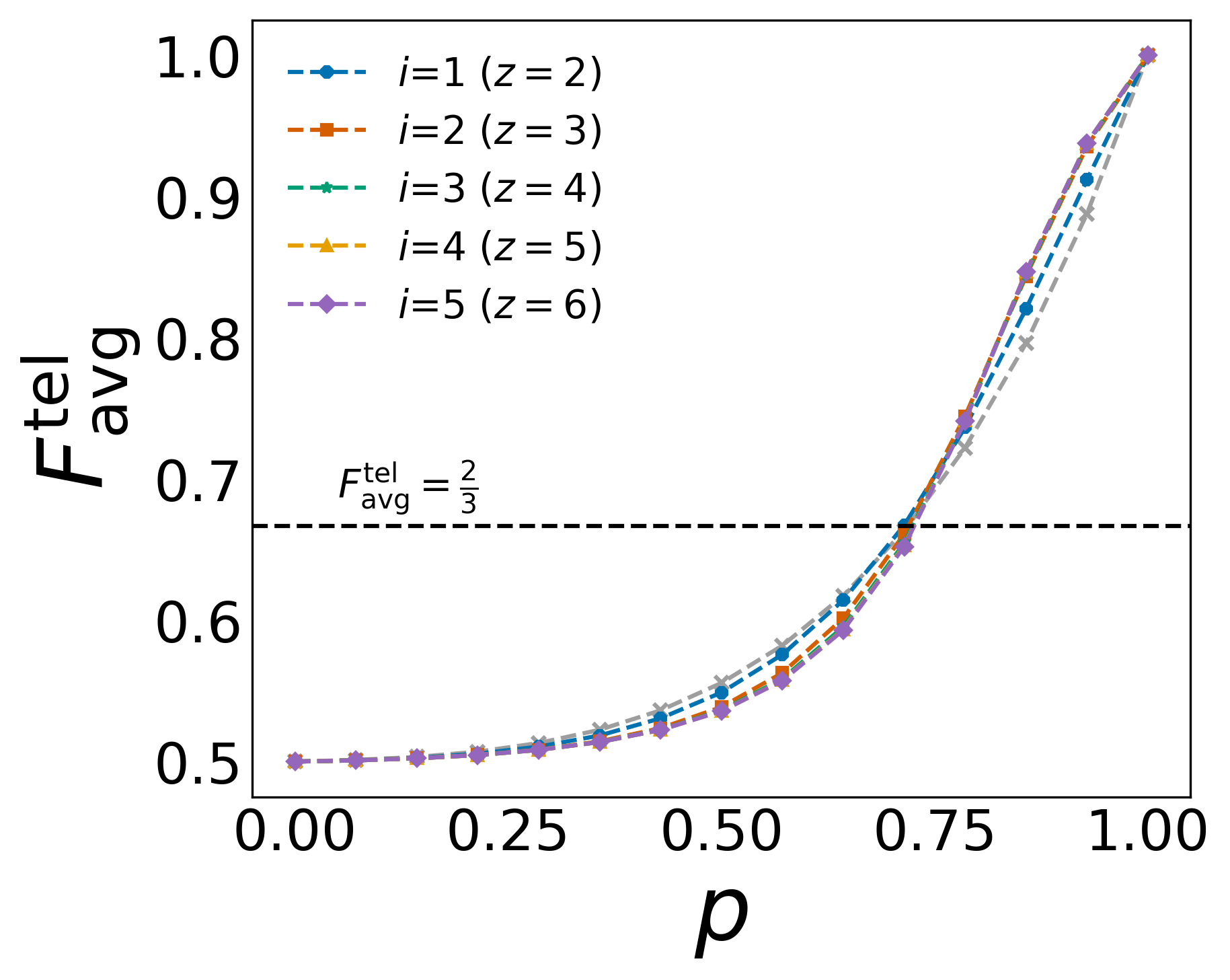}
    \includegraphics[scale=0.27]
        {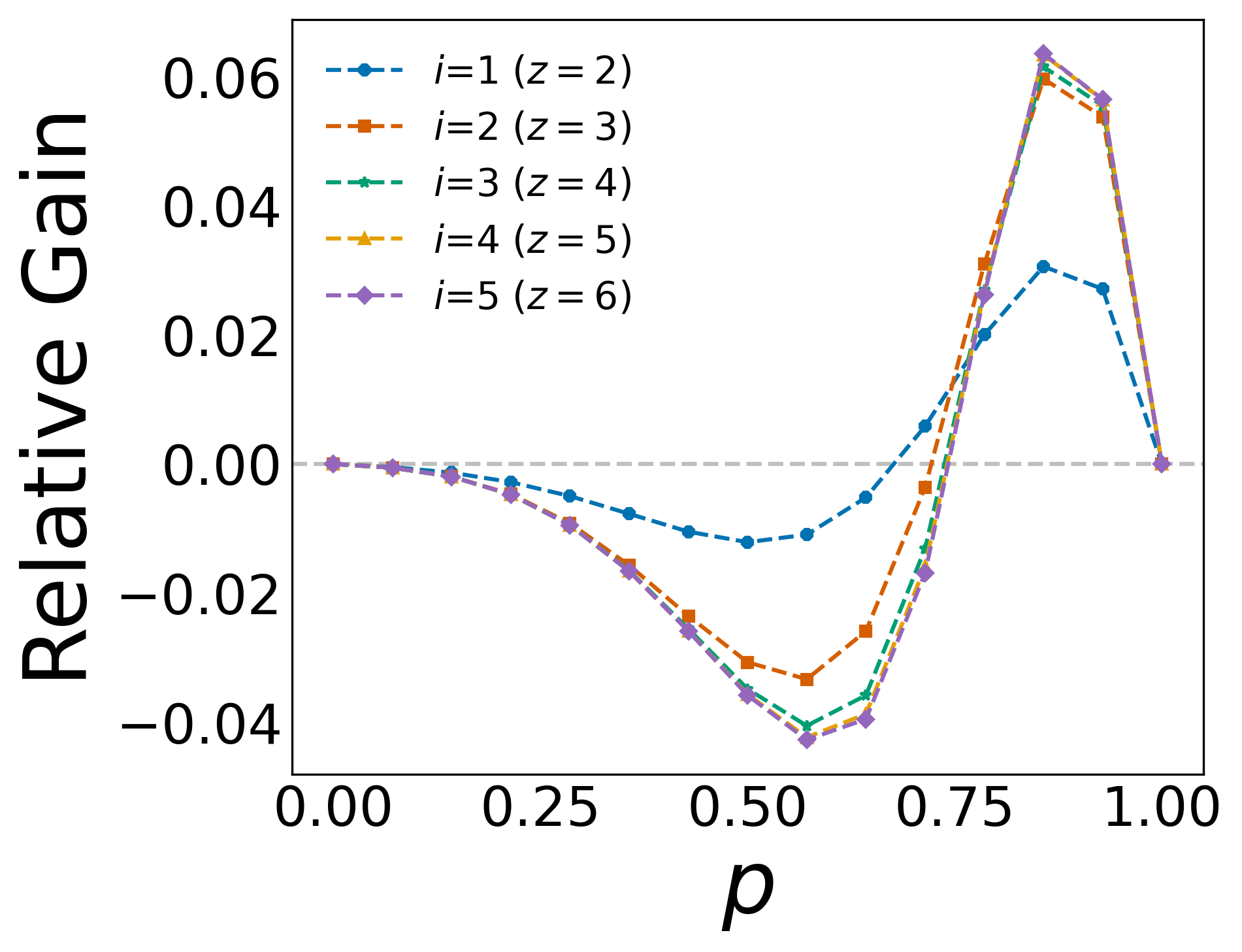}
    \\[1ex]
    \hspace*{0.75cm}(a) \hspace{3.9cm} (b)
    \caption{Plots of (a) $F^{\mathrm{tel}}_{\mathrm{avg}}$ as a function of $p$ and (b) relative gain as a function of $p$ for an Erd\H{o}s-R\'enyi network with $N=100$ and $E=250$ where $k=1$. The horizontal black line suggests $\Ftel=\frac{2}{3}$. The grey line denotes the variation of $\Ftel$ with respect to $p$ for the without MPEP scenario.}
    \label{Fvsp_ERN_k1}
\end{figure}\\
\indent For $k=k_{\max}$, we can obtain the improvement for the same ERN at $p \approx 0.7135, 0.7176, 0.7237, 0.7259,$ and $ 0.7277$ for $z = 2, 3, 4, 5$ and $6$. Fig. \ref{Fvsp_ERN_kmax}(a) and \ref{Fvsp_ERN_kmax}(b) illustrate these observations, as well as the variation of relative gains with the $p$. These findings indicate that the proposed techniques improve $\Ftel$ from the without MPEP method.  %
\begin{figure}[h]
    \centering
    \includegraphics[scale=0.27]
        {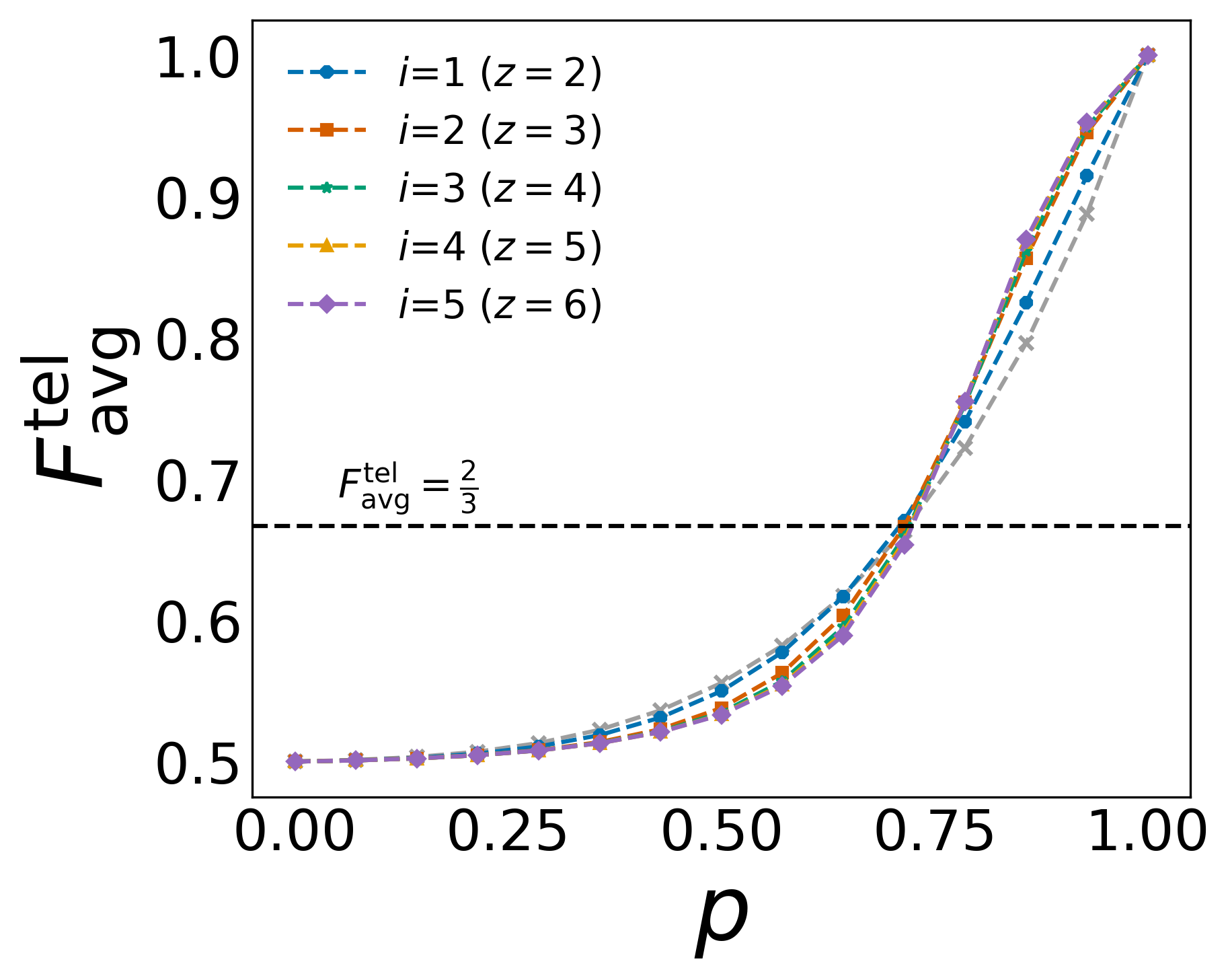}
    \includegraphics[scale=0.27]
        {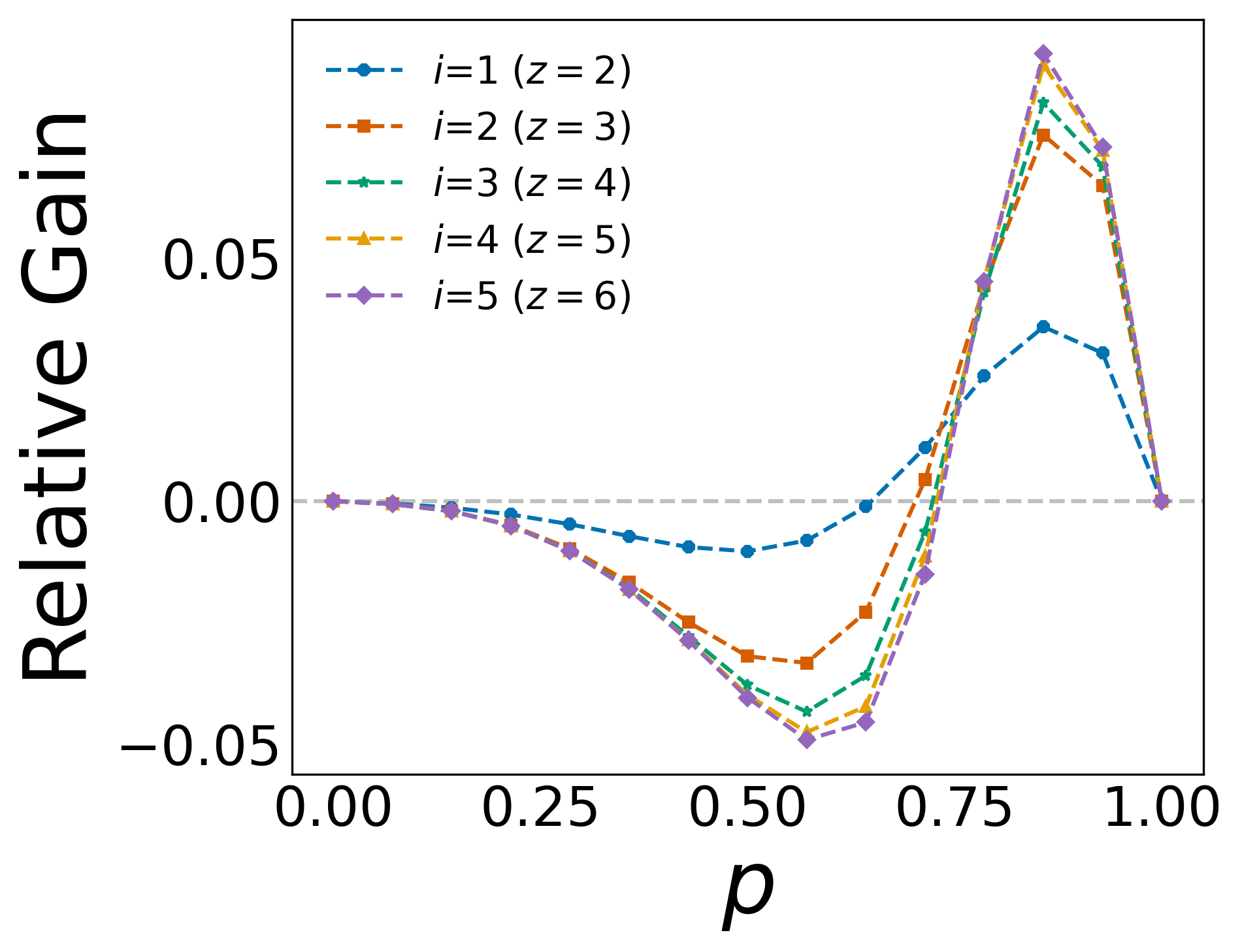}
    \\[1ex]
    \hspace*{0.75cm}(a) \hspace{3.9cm} (b)
    \caption{Plots of (a) $F^{\mathrm{tel}}_{\mathrm{avg}}$ as a function of $p$ and (b) relative gain as a function of $p$ for an Erd\H{o}s-R\'enyi network with $N=100$ and $E=250$ where $k=k_{\mathrm{max}}$. The horizontal black line suggests $\Ftel=\frac{2}{3}$. The grey line denotes the variation of $\Ftel$ with respect to $p$ for the without MPEP scenario.}
    \label{Fvsp_ERN_kmax}
\end{figure}\\
\indent In the case of ERN, the $\Ftel$ approaches saturation after a certain number of $z$, where $i_{\max}=z-1$, similar to the regular topologies. Additionally, as $z$ increases, the paths with larger pathlengths appear in the path set. As a result, $\Ftel$ decreases gradually. However, we discovered that the change in $\Ftel$ is of the order of $10^{-4}$ after $z=6$ and $i=5=i_{\max}$ when $k=1$, and $z=7$ and $i=6=i_{\max}$ for $k=k_{\max}$. These results are illustrated in Fig. \ref{Fvsp_ERN_k1}(a) and Fig. \ref{Fvsp_ERN_kmax}. Hence, we may conclude that increasing the quantity of $z$ and $i$ has no substantial effect on $\Ftel$.\\

\section{Conclusion}
\label{sec6}
In this article, we propose two different algorithms to evaluate the average teleportation fidelity, $\Ftel$, based on two MPEP strategies. We show that this approach improves the value of $\Ftel$ by calculating the analytical expression for the topologies with loops, including the ring network, the complete graph network, the triangular lattice network, and the square lattice network. We further examine these findings for both the irregular topology of the Erd\H{o}s–R\'enyi network and the previously discussed regular topologies.
Furthermore, we prove that the shortest path last (SPL) strategy leads to a better value of $\Ftel$ than the shortest path first (SPF) strategy. Considering all the cases, we find that the $\Ftel$ is greater than that obtained without the MPEP approach for appropriate values of the isotropic state parameter $p$ and a finite number of purification stages. We also show the variation of the associated relative gain with $p$. Moreover, for lattice networks, the proposed algorithms can achieve the quantum advantage for a relatively lower value of $p_c$ than without MPEP method.  When the network has a limited number of entangled states as resources between two nodes, the edges can not be used multiple times for different paths. We also show that the calculation of $\Ftel$ without MPEP strategy can be outperformed by using a single edge only once, along with a suitable purification approach for a CGN. For TLN and SLN, increasing the number of purifications along with the total number of paths does not enhance the $\Ftel$ indefinitely, as it reaches saturation after a certain number of steps. On the other hand, we observe the saturation of $\Ftel$ in the case of ERN also. These findings show that even in the absence of resources, the MPEP approach yields better fidelity than without MPEP about the effective distribution of resources in quantum networks.
These findings suggest that MPEP techniques can greatly enhance quantum network topologies' capacity to transfer quantum information. 

\section*{Acknowledgement}
S.R. would like to acknowledge DST-India for the INSPIRE Fellowship (IF220695) support.

\bibliography{ref.bib}
\newpage
\onecolumngrid
\appendix

\section{Derivation of the formula of maximum teleportation fidelity $F^{\mathrm{tel}}_{\mathrm{max}}$}
\label{appA}

Let's begin with the representation of the state $\rho$ in the Bell basis, which acts on the Hilbert space $\mathcal{H}=\mathcal{H}_1 \otimes \mathcal{H}_2$.  
\begin{eqnarray}
    \rho = A \dyad{\phi^+}+B \dyad{\phi^-}+C \dyad{\psi^+}+D \dyad{\psi^-},
    \label{eq:A1}
\end{eqnarray}
where  $\ket{\phi^\pm} = \frac{1}{\sqrt{2}} = \ket{00} \pm \ket{11}$, $\ket{\psi^\pm} = \frac{1}{\sqrt{2}} = \ket{01} \pm \ket{10}$. From \cite{horodecki96}, the maximum teleportation fidelity is given by
\begin{eqnarray}
    F^{\mathrm{tel}}_{\mathrm{max}} = \frac{1}{2}\lrfb{1+\frac{1}{3}\Tr\sqrt{T^\dagger T}},
    \label{eq:A2}
\end{eqnarray}
where the entries of $T$ matrix is given by $t_{ij} = \Tr(\rho \sigma_i \otimes \sigma_j)$. For the state Eq. \eqref{eq:A1}, the matrix is $T = \diag(A-B+C-D, -A+B+C-D, A+B-C-D) = T^\dagger$ which leads to $\sqrt{T^\dagger T} = \diag(A-B+C-D, A-B-C+D, A+B-C-D)$ and hence $\Tr\sqrt{T^\dagger T} = 3A-B-C-D$. Finally, from Eq. \eqref{eq:A2},
\begin{eqnarray}
    F^{\mathrm{tel}}_{\mathrm{max}} = \frac{1}{2}\lrtb{1+\frac{1}{3}(3A-B-C-D)}.
    \label{eq:A3}
\end{eqnarray}
In case of a bipartite isotropic state $\rho = \frac{1-p}{4}\mathbb{I}+p\dyad{\Phi^+} = \frac{1+3p}{4}\dyad{\Phi^+}+\frac{1-p}{4}(\dyad{\Phi^-}+\dyad{\Psi^+}+\dyad{\Psi^-})$, the value of maximum teleportation fidelity is
\begin{eqnarray}
    F^{\mathrm{tel}}_{\mathrm{max}} = \frac{1}{2}\lrtb{1+\frac{1}{3}(3A-B-C-D)} = \frac{1}{2}(1+p).
\end{eqnarray}

\end{document}